\documentclass[dvipsnames, fleqn, 11pt]{article}
\usepackage{amsmath}
\usepackage{amsthm}
\usepackage{natbib}

\usepackage{classicthesis}

\areaset[current]{420pt}{761pt} 
\usepackage{url}
\usepackage{tikz}
\usepackage{chngcntr} 
\usepackage{caption}
\usepackage{subcaption}
\usepackage{scalerel}

\makeatletter
\@ifundefined{lhs2tex.lhs2tex.sty.read}%
  {\@namedef{lhs2tex.lhs2tex.sty.read}{}%
   \newcommand\SkipToFmtEnd{}%
   \newcommand\EndFmtInput{}%
   \long\def\SkipToFmtEnd#1\EndFmtInput{}%
  }\SkipToFmtEnd

\newcommand\ReadOnlyOnce[1]{\@ifundefined{#1}{\@namedef{#1}{}}\SkipToFmtEnd}
\usepackage{amstext}
\usepackage{amssymb}
\usepackage{stmaryrd}
\DeclareFontFamily{OT1}{cmtex}{}
\DeclareFontShape{OT1}{cmtex}{m}{n}
  {<5><6><7><8>cmtex8
   <9>cmtex9
   <10><10.95><12><14.4><17.28><20.74><24.88>cmtex10}{}
\DeclareFontShape{OT1}{cmtex}{m}{it}
  {<-> ssub * cmtt/m/it}{}

\DeclareFontShape{OT1}{cmtt}{bx}{n}
  {<5><6><7><8>cmtt8
   <9>cmbtt9
   <10><10.95><12><14.4><17.28><20.74><24.88>cmbtt10}{}
\DeclareFontShape{OT1}{cmtex}{bx}{n}
  {<-> ssub * cmtt/bx/n}{}

\newcommand{\Conid}[1]{\mathit{#1}}
\newcommand{\Varid}[1]{\mathit{#1}}
\newcommand{\anonymous}{\kern0.06em \vbox{\hrule\@width.5em}}

\renewcommand{\leq}{\leqslant}
\renewcommand{\geq}{\geqslant}
\usepackage{polytable}

\@ifundefined{mathindent}%
  {\newdimen\mathindent\mathindent\leftmargini}%
  {}%

\def\resethooks{%
  \global\let\SaveRestoreHook\empty
  \global\let\ColumnHook\empty}
\newcommand*{\savecolumns}[1][default]%
  {\g@addto@macro\SaveRestoreHook{\savecolumns[#1]}}
\newcommand*{\restorecolumns}[1][default]%
  {\g@addto@macro\SaveRestoreHook{\restorecolumns[#1]}}
\newcommand*{\aligncolumn}[2]%
  {\g@addto@macro\ColumnHook{\column{#1}{#2}}}

\resethooks

\newcommand{\onelinecommentchars}{\quad-{}- }
\newcommand{\commentbeginchars}{\enskip\{-}
\newcommand{\commentendchars}{-\}\enskip}

\newcommand{\visiblecomments}{%
  \let\onelinecomment=\onelinecommentchars
  \let\commentbegin=\commentbeginchars
  \let\commentend=\commentendchars}

\newcommand{\invisiblecomments}{%
  \let\onelinecomment=\empty
  \let\commentbegin=\empty
  \let\commentend=\empty}

\visiblecomments

\newlength{\blanklineskip}
\newcommand{\hsindent}[1]{\quad}
\let\hspre\empty
\let\hspost\empty

\EndFmtInput
\makeatother
\ReadOnlyOnce{polycode.fmt}%
\makeatletter

\newcommand{\hsnewpar}[1]%
  {{\parskip=0pt\parindent=0pt\par\vskip #1\noindent}}

\newcommand{\hscodestyle}{}

\newcommand{\sethscode}[1]%
  {\expandafter\let\expandafter\hscode\csname #1\endcsname
   \expandafter\let\expandafter\endhscode\csname end#1\endcsname}

  {\par\noindent
   \advance\leftskip\mathindent
   \hscodestyle
   \let\\=\@normalcr
   \let\hspre\(\let\hspost\)%
   \pboxed}%
  {\endpboxed\)%
   \par\noindent
   \ignorespacesafterend}

  {\hsnewpar\abovedisplayskip
   \advance\leftskip\mathindent
   \hscodestyle
   \let\hspre\(\let\hspost\)%
   \pboxed}%
  {\endpboxed%
   \hsnewpar\belowdisplayskip
   \ignorespacesafterend}

  {\hsnewpar\abovedisplayskip
   \advance\leftskip\mathindent
   \hscodestyle
   \let\\=\@normalcr
   \(\pboxed}%
  {\endpboxed\)%
   \hsnewpar\belowdisplayskip
   \ignorespacesafterend}

\newcommand{\plainhs}{\sethscode{plainhscode}}

\plainhs

  {\hsnewpar\abovedisplayskip
   \advance\leftskip\mathindent
   \hscodestyle
   \let\\=\@normalcr
   \(\parray}%
  {\endparray\)%
   \hsnewpar\belowdisplayskip
   \ignorespacesafterend}

  {\parray}{\endparray}

  {\(\parray}{\endparray\)}

\def\codeframewidth{\arrayrulewidth}
\RequirePackage{calc}

  {\parskip=\abovedisplayskip\par\noindent
   \hscodestyle
   \arrayrulewidth=\codeframewidth
   \tabular{@{}|p{\linewidth-2\arraycolsep-2\arrayrulewidth-2pt}|@{}}%
   \hline\framedhslinecorrect\\{-1.5ex}%
   \let\endoflinesave=\\
   \let\\=\@normalcr
   \(\pboxed}%
  {\endpboxed\)%
   \framedhslinecorrect\endoflinesave{.5ex}\hline
   \endtabular
   \parskip=\belowdisplayskip\par\noindent
   \ignorespacesafterend}

\newcommand{\framedhslinecorrect}[2]%
  {#1[#2]}

  {\(\def\column##1##2{}%
   \let\>\undefined\let\<\undefined\let\\\undefined
   \newcommand\>[1][]{}\newcommand\<[1][]{}\newcommand\\[1][]{}%
   \def\fromto##1##2##3{##3}%
   }{\) }%

  {\let\orighscode=\hscode
   \let\origendhscode=\endhscode
   \def\endhscode{\def\hscode{\endgroup\def\@currenvir{hscode}\\}\begingroup}
   \orighscode\def\hscode{\endgroup\def\@currenvir{hscode}}}%
  {\origendhscode
   \global\let\hscode=\orighscode
   \global\let\endhscode=\origendhscode}%

\makeatother
\EndFmtInput
\ReadOnlyOnce{forall.fmt}%
\makeatletter

\let\HaskellResetHook\empty
\newcommand*{\AtHaskellReset}[1]{%
  \g@addto@macro\HaskellResetHook{#1}}
\newcommand*{\HaskellReset}{\HaskellResetHook}

\newcommand\hsforall{\global\let\hsdot=\hsperiodonce}
\newcommand*\hsperiodonce[2]{#2\global\let\hsdot=\hscompose}
\newcommand*\hscompose[2]{#1}

\AtHaskellReset{\global\let\hsdot=\hscompose}

\HaskellReset

\makeatother
\EndFmtInput
\ReadOnlyOnce{exists.fmt}%
\makeatletter

\newcommand\hsexists{\global\let\hsdot=\hsperiodonce}

\AtHaskellReset{\global\let\hsdot=\hscompose}

\HaskellReset

\makeatother
\EndFmtInput

\newcommand{\circo}{\mathrel{\kern 0.12em%
      \raisebox{1pt}{\tikz \draw[line width=0.6pt] circle(1.1pt);}%
      \kern 0.12em}}

\newlength{\mylen}
\setbox1=\hbox{$\bullet$}\setbox2=\hbox{\tiny$\bullet$}
\newcommand{\myapply}{\mathrel{\kern 0.12em\scalebox{0.8}{\$}\kern 0.12em}}

\renewcommand{\Conid}[1]{{\mathsf{#1}}}

\usepackage{doubleequals}

\newtheorem{theorem}{Theorem}
\newtheorem{lemma}{Lemma}

\def\commentbegin{\quad\begingroup\color{SeaGreen}\{\ }
\def\commentend{\}\endgroup}

\mathindent=\parindent
\definecolor{mediumpersianblue}{rgb}{0.0, 0.4, 0.65}
\everymath{\color{mediumpersianblue}}
\apptocmd{\[}{\color{mediumpersianblue}}{}{}
\AtBeginEnvironment{equation}{\color{mediumpersianblue}}
\AtBeginEnvironment{equation*}{\color{mediumpersianblue}}

\allowdisplaybreaks

\newcommand\numberthis{\refstepcounter{equation}\tag{\theequation}}

\counterwithout{equation}{section}

\begin{document}

\title{Functional Pearl:\\
Bottom-Up Computation Using Trees of Sublists}

\author{\color{black}Shin-Cheng Mu}
\date{%
Institute of Information Science, Academia Sinica
}

\maketitle

\begin{abstract}
Some top-down problem specifications, if executed directly, may compute sub-problems repeatedly.
Instead, we may want a bottom-up algorithm that stores solutions of sub-problems in a table to be reused.
It can be tricky, however, to figure out how the table can be represented and efficiently maintained.
We study a special case: computing a function \ensuremath{\Varid{h}} taking lists as inputs such that \ensuremath{\Varid{h}\;\Varid{xs}} is defined in terms of all immediate sublists of \ensuremath{\Varid{xs}}.
Richard Bird studied this problem in 2008, and presented a concise but cryptic algorithm without much explanation.
We give this algorithm a proper derivation, and discover a key property that allows it to work.
The algorithm builds trees that have certain shapes --- the sizes along the left spine is a diagonal in Pascal's triangle.
The crucial function we derive transforms one diagonal to the next.
\end{abstract}

\section{Introduction}

A list \ensuremath{\Varid{ys}} is said to be an \emph{immediate sublist} of \ensuremath{\Varid{xs}} if \ensuremath{\Varid{ys}} can be obtained by removing exactly one element from \ensuremath{\Varid{xs}}.
For example, the four immediate sublists of \ensuremath{\text{\ttfamily \char34 abcd\char34}} are \ensuremath{\text{\ttfamily \char34 abc\char34}}, \ensuremath{\text{\ttfamily \char34 abd\char34}}, \ensuremath{\text{\ttfamily \char34 acd\char34}}, and \ensuremath{\text{\ttfamily \char34 bcd\char34}}.
Consider computing a function \ensuremath{\Varid{h}} that takes a list as input, with the property that the value of \ensuremath{\Varid{h}\;\Varid{xs}} depends on values of \ensuremath{\Varid{h}} at all the immediate sublists of \ensuremath{\Varid{xs}}.
For example, as seen in Figure \ref{fig:td-call-tree}, \ensuremath{\Varid{h}\;\text{\ttfamily \char34 abcd\char34}} depends on \ensuremath{\Varid{h}\;\text{\ttfamily \char34 abc\char34}}, \ensuremath{\Varid{h}\;\text{\ttfamily \char34 abd\char34}}, \ensuremath{\Varid{h}\;\text{\ttfamily \char34 acd\char34}}, and \ensuremath{\Varid{h}\;\text{\ttfamily \char34 bcd\char34}}.
In this top-down manner, to compute \ensuremath{\Varid{h}\;\text{\ttfamily \char34 abc\char34}} we make calls to \ensuremath{\Varid{h}\;\text{\ttfamily \char34 ab\char34}}, \ensuremath{\Varid{h}\;\text{\ttfamily \char34 ac\char34}}, and \ensuremath{\Varid{h}\;\text{\ttfamily \char34 bc\char34}}; to compute \ensuremath{\Varid{h}\;\text{\ttfamily \char34 abd\char34}}, we make a call to \ensuremath{\Varid{h}\;\text{\ttfamily \char34 ab\char34}} as well --- many values end up being re-computed.
One would like to instead proceed in a bottom-up manner, storing computed values so that they can be reused.
For this problem, one might want to build a lattice-like structure, like that in Figure~\ref{fig:ch-lattice}, from bottom to top,
sharing the value on one layer that are used in constructing the next layer.

\begin{figure}[h]
\centering
\includegraphics[width=0.5\textwidth]{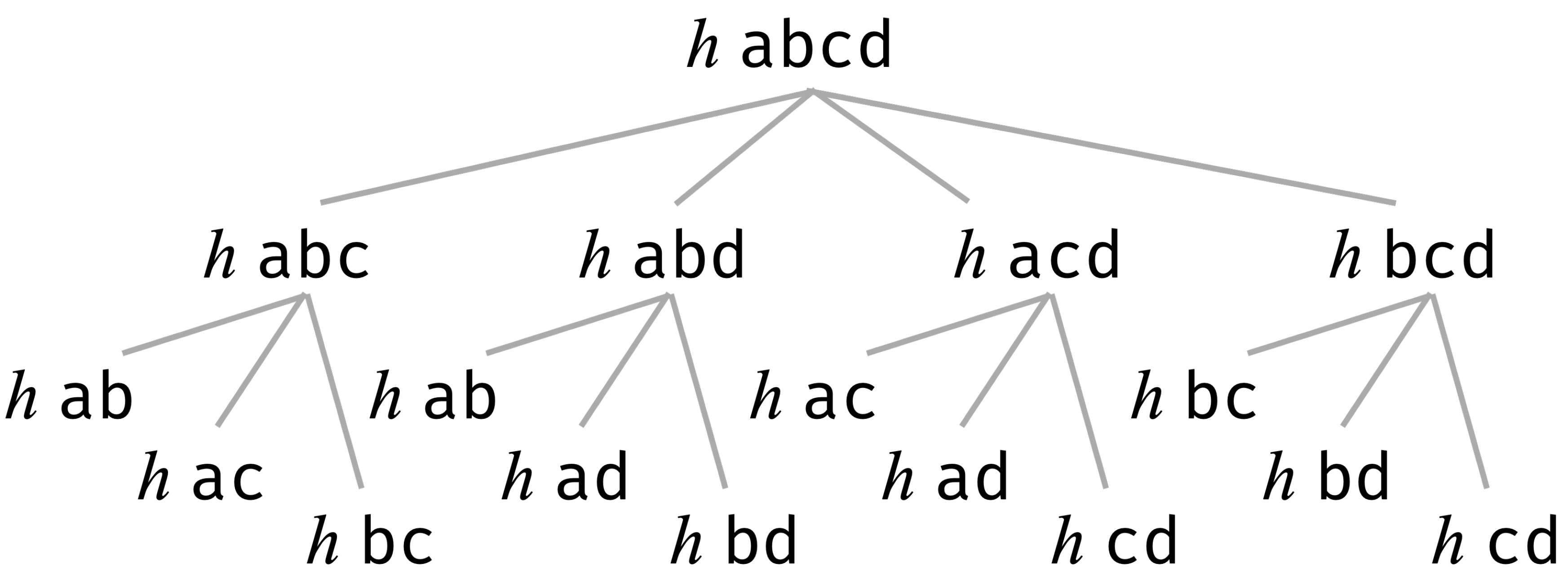}
\caption{Computing \ensuremath{\Varid{h}\;\text{\ttfamily \char34 abcd\char34}} top-down. String constants are shown using monospace font but without quotes, to save space.}
\label{fig:td-call-tree}
\end{figure}

\cite{Bird:08:Zippy} presented an interesting study of the relationship between top-down and bottom-up algorithms.
It was shown that if an algorithm can be written in a specific top-down style,
with ingredients that satisfy certain properties,
there is an equivalent bottom-up algorithm that stores intermediate results in a table.
The ``all immediate sublists'' instance was the last example of the paper.
To satisfy the said properties, however, Bird had to introduce additional data structures and helper functions out of the blue.
The rationale for designing these data structures and functions was not obvious, nor was it clear why the needed properties are met.
The resulting bottom-up algorithm is concise, elegant, but also cryptic --- all the more reason to present the proper calculation it deserves.

\begin{figure}
\centering
\includegraphics[width=0.85\textwidth]{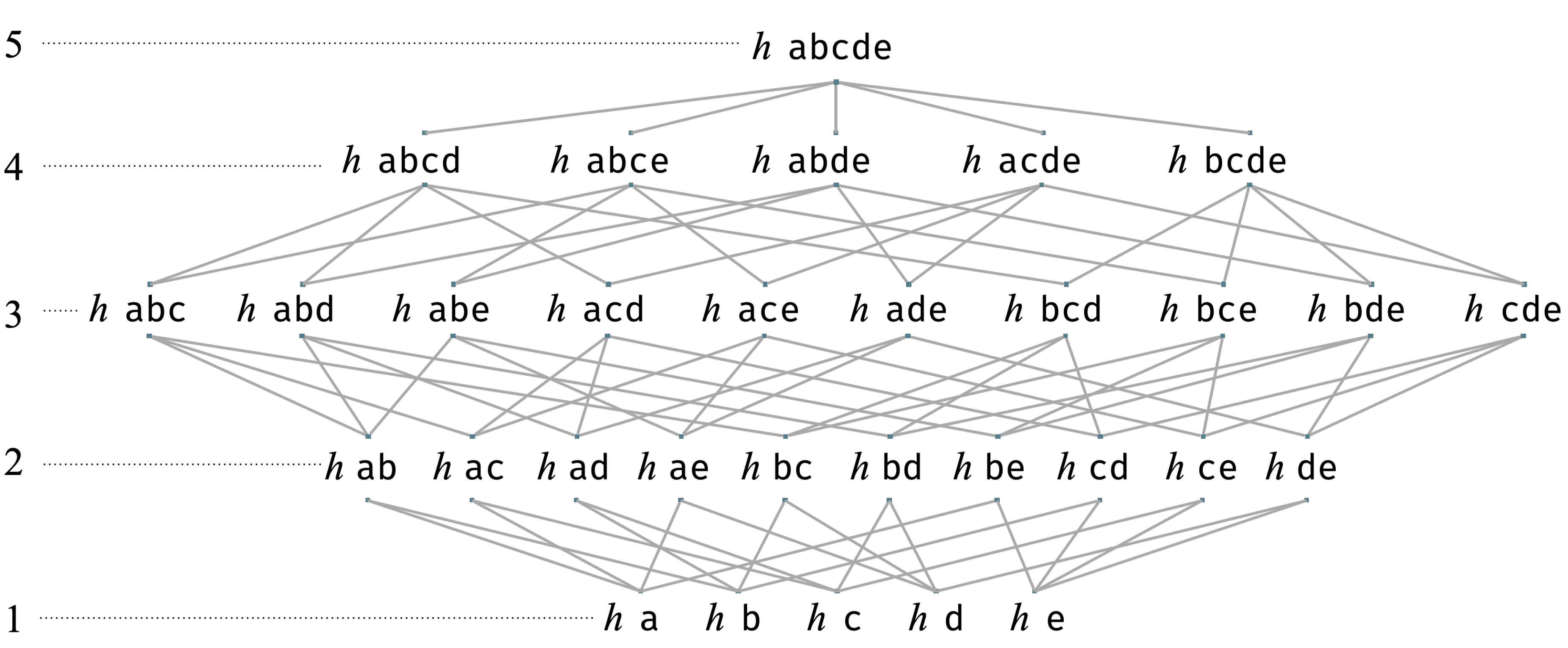}
\caption{Computing \ensuremath{\Varid{h}\;\text{\ttfamily \char34 abcde\char34}} bottom-up.}
\label{fig:ch-lattice}
\end{figure}

In this pearl we review this problem, present an alternative specification, and derive Bird's algorithm.
It turns out that the key property we rely on is different from that in Bird's \citeyearpar{Bird:08:Zippy}.
Driven by this property, our main derivation is much more straight-forward.
This suggests that, while many bottom-up algorithms look alike, the reasons why they work may be more diverse than we thought, and there is a lot more to be discovered regarding reasoning about their correctness.

Before we start, one might ask: are there actually problems whose solution of input \ensuremath{\Varid{xs}} depends only on solutions for immediate sublists of \ensuremath{\Varid{xs}}?
It turns out that it is quite common.
While problems such as \emph{minimum editing distance} or \emph{longest common subsequence} are defined on two lists,
it is known in the algorithm community that, with clever encoding, they can be rephrased as problems defined on one list,
whose solution depends on immediate sublists.
Many problems in additive combinatorics \citep{TaoVu:12:Additive} can also be cast into this form.

\section{Specification}
\label{sec:spec}

We use a Haskell-like notation throughout the paper.
Like in Haskell, if a function is defined by multiple clauses, the patterns and guards are matched top to bottom.
Differences from Haskell include that we allow \ensuremath{\Varid{n}\mathbin{+}\Varid{k}} patterns, and that we denote the type of lists by \ensuremath{\Conid{L}}.
Since we will use natural transformations and \ensuremath{\Varid{map}} a lot, for brevity we denote the \ensuremath{\Varid{map}} function of lists as \ensuremath{\Varid{L}\mathbin{::}(\Varid{a}\to \Varid{b})\to \Conid{L}\;\Varid{a}\to \Conid{L}\;\Varid{b}} (note that \ensuremath{\Varid{L}} is written in italic font, to distinguish it from the type constructor \ensuremath{\Conid{L}}).

The immediate sublists of a list can be specified in many ways.
We use the definition below mainly because it generates sublists in an intuivite order:
\begin{hscode}\SaveRestoreHook
\column{B}{@{}>{\hspre}l<{\hspost}@{}}%
\column{14}{@{}>{\hspre}l<{\hspost}@{}}%
\column{E}{@{}>{\hspre}l<{\hspost}@{}}%
\>[B]{}\Varid{subs}\mathbin{::}\Conid{L}\;\Varid{a}\to \Conid{L}\;(\Conid{L}\;\Varid{a}){}\<[E]%
\\
\>[B]{}\Varid{subs}\;[\mskip1.5mu \mskip1.5mu]{}\<[14]%
\>[14]{}\mathrel{=}[\mskip1.5mu \mskip1.5mu]{}\<[E]%
\\
\>[B]{}\Varid{subs}\;(\Varid{x}\mathbin{:}\Varid{xs}){}\<[14]%
\>[14]{}\mathrel{=}\Varid{L}\;(\Varid{x}\mathbin{:})\;(\Varid{subs}\;\Varid{xs})\mathbin{{+}\mskip-8mu{+}}[\mskip1.5mu \Varid{xs}\mskip1.5mu]~~.{}\<[E]%
\ColumnHook
\end{hscode}\resethooks
For example, \ensuremath{\Varid{subs}\;\text{\ttfamily \char34 abcde\char34}} yields \ensuremath{[\mskip1.5mu \text{\ttfamily \char34 abcd\char34},\text{\ttfamily \char34 abce\char34},\text{\ttfamily \char34 abde\char34},\text{\ttfamily \char34 acde\char34},\text{\ttfamily \char34 bcde\char34}\mskip1.5mu]}.

Denote the function we wish to compute by \ensuremath{\Varid{h}\mathbin{::}\Conid{L}\;\Conid{X}\to \Conid{Y}} for some types \ensuremath{\Conid{X}} and \ensuremath{\Conid{Y}}.
We assume that it is a partial function defined on non-empty lists, and can be computed top-down as below:
\begin{hscode}\SaveRestoreHook
\column{B}{@{}>{\hspre}l<{\hspost}@{}}%
\column{8}{@{}>{\hspre}l<{\hspost}@{}}%
\column{E}{@{}>{\hspre}l<{\hspost}@{}}%
\>[B]{}\Varid{h}\mathbin{::}\Conid{L}\;\Conid{X}\to \Conid{Y}{}\<[E]%
\\
\>[B]{}\Varid{h}\;[\mskip1.5mu \Varid{x}\mskip1.5mu]{}\<[8]%
\>[8]{}\mathrel{=}\Varid{f}\;\Varid{x}{}\<[E]%
\\
\>[B]{}\Varid{h}\;\Varid{xs}{}\<[8]%
\>[8]{}\mathrel{=}\Varid{g}\circo\Varid{L}\;\Varid{h}\circo\Varid{subs}\myapply\Varid{xs}~~,{}\<[E]%
\ColumnHook
\end{hscode}\resethooks
where \ensuremath{\Varid{f}\mathbin{::}\Conid{X}\to \Conid{Y}} is used in the base case when the input is a singleton list,
and \ensuremath{\Varid{g}\mathbin{::}\Conid{L}\;\Conid{Y}\to \Conid{Y}} is for the inductive case. The operator \ensuremath{(\myapply)} denotes function application, which binds looser than function composition \ensuremath{(\circo)}.
We sometimes use \ensuremath{(\myapply)} to reduce the number of parentheses.

For this pearl, it is convenient to use an equivalent definition.
Let \ensuremath{\Varid{td}} be a family of functions indexed by natural numbers (denoted by \ensuremath{\Conid{Nat}}):
\begin{hscode}\SaveRestoreHook
\column{B}{@{}>{\hspre}l<{\hspost}@{}}%
\column{11}{@{}>{\hspre}l<{\hspost}@{}}%
\column{E}{@{}>{\hspre}l<{\hspost}@{}}%
\>[B]{}\Varid{td}\mathbin{::}\Conid{Nat}\to \Conid{L}\;\Conid{X}\to \Conid{Y}{}\<[E]%
\\
\>[B]{}\Varid{td}\;\mathrm{0}{}\<[11]%
\>[11]{}\mathrel{=}\Varid{f}\circo\Varid{ex}{}\<[E]%
\\
\>[B]{}\Varid{td}\;(\mathrm{1}\mathbin{+}\Varid{n}){}\<[11]%
\>[11]{}\mathrel{=}\Varid{g}\circo\Varid{L}\;(\Varid{td}\;\Varid{n})\circo\Varid{subs}~~,{}\<[E]%
\ColumnHook
\end{hscode}\resethooks
Here the function \ensuremath{\Varid{ex}\mathbin{::}\Conid{L}\;\Varid{a}\to \Varid{a}} takes a singleton list and extracts the only component.
The intention is that \ensuremath{\Varid{td}\;\Varid{n}} is a function defined on lists of length exactly \ensuremath{\mathrm{1}\mathbin{+}\Varid{n}}.
Given input \ensuremath{\Varid{xs}}, the value we aim to compute is \ensuremath{\Varid{h}\;\Varid{xs}\mathrel{=}\Varid{td}\;(\Varid{length}\;\Varid{xs}\mathbin{-}\mathrm{1})\;\Varid{xs}}.
This definition will be handy later.

The function \ensuremath{\Varid{iter}\;\Varid{k}} composes a function with itself \ensuremath{\Varid{k}} times:
\begin{hscode}\SaveRestoreHook
\column{B}{@{}>{\hspre}l<{\hspost}@{}}%
\column{13}{@{}>{\hspre}l<{\hspost}@{}}%
\column{E}{@{}>{\hspre}l<{\hspost}@{}}%
\>[B]{}\Varid{iter}\mathbin{::}\Conid{Nat}\to (\Varid{a}\to \Varid{a})\to \Varid{a}\to \Varid{a}{}\<[E]%
\\
\>[B]{}\Varid{iter}\;\mathrm{0}\;{}\<[13]%
\>[13]{}\Varid{f}\mathrel{=}\Varid{id}{}\<[E]%
\\
\>[B]{}\Varid{iter}\;(\mathrm{1}\mathbin{+}\Varid{k})\;{}\<[13]%
\>[13]{}\Varid{f}\mathrel{=}\Varid{iter}\;\Varid{k}\;\Varid{f}\circo\Varid{f}~~.{}\<[E]%
\ColumnHook
\end{hscode}\resethooks
For brevity, we will write \ensuremath{\Varid{iter}\;\Varid{k}\;\Varid{f}} as \ensuremath{{\Varid{f}}^{\Varid{k}}} for the rest of this pearl.
The bottom-up algorithm we aim to construct has the following form:
\begin{hscode}\SaveRestoreHook
\column{B}{@{}>{\hspre}l<{\hspost}@{}}%
\column{E}{@{}>{\hspre}l<{\hspost}@{}}%
\>[B]{}\Varid{bu}\mathbin{::}\Conid{Nat}\to \Conid{L}\;\Conid{X}\to \Conid{Y}{}\<[E]%
\\
\>[B]{}\Varid{bu}\;\Varid{n}\mathrel{=}\Varid{post}\circo{\Varid{step}}^{\Varid{n}}\circo\Varid{pre}~~,{}\<[E]%
\ColumnHook
\end{hscode}\resethooks
where \ensuremath{\Varid{pre}} preprocesses the input and builds the lowest level in Figure \ref{fig:ch-lattice},
and each \ensuremath{\Varid{step}} builds a level from the one below.
For input of length \ensuremath{\mathrm{1}\mathbin{+}\Varid{n}} we repeat \ensuremath{\Varid{n}} times and, by then, we can extract the singleton value by \ensuremath{\Varid{post}}.

The aim of this pearl is to construct \ensuremath{\Varid{pre}}, \ensuremath{\Varid{step}}, and \ensuremath{\Varid{post}} such that \ensuremath{\Varid{td}\mathrel{=}\Varid{bu}}.


\section{Building a New Level}
\label{sec:build-level}

\newcommand{\dqoute}{\mathtt{"}}

To find out what \ensuremath{\Varid{step}} might be, we need to figure out how to specify a level, and what happens when a level is built from the one below it.
We use Figure~\ref{fig:ch-lattice} as our motivating example.
As one can see, level \ensuremath{\mathrm{2}} in Figure~\ref{fig:ch-lattice} consists of sublists of \ensuremath{\mathtt{\dqoute abcde\dqoute}} that have length \ensuremath{\mathrm{2}}, and level \ensuremath{\mathrm{3}} consists of sublists having length \ensuremath{\mathrm{3}}, and so on.
Let \ensuremath{\Varid{choose}\;\Varid{k}\;\Varid{xs}} denote choosing \ensuremath{\Varid{k}} elements from the list \ensuremath{\Varid{xs}}:
\begin{hscode}\SaveRestoreHook
\column{B}{@{}>{\hspre}l<{\hspost}@{}}%
\column{15}{@{}>{\hspre}l<{\hspost}@{}}%
\column{23}{@{}>{\hspre}c<{\hspost}@{}}%
\column{23E}{@{}l@{}}%
\column{26}{@{}>{\hspre}l<{\hspost}@{}}%
\column{E}{@{}>{\hspre}l<{\hspost}@{}}%
\>[B]{}\Varid{choose}\mathbin{::}\Conid{Nat}\to \Conid{L}\;\Varid{a}\to \Conid{L}\;(\Conid{L}\;\Varid{a}){}\<[E]%
\\
\>[B]{}\Varid{choose}\;\mathrm{0}\;{}\<[15]%
\>[15]{}\anonymous {}\<[23]%
\>[23]{}\mathrel{=}{}\<[23E]%
\>[26]{}[\mskip1.5mu [\mskip1.5mu \mskip1.5mu]\mskip1.5mu]{}\<[E]%
\\
\>[B]{}\Varid{choose}\;\Varid{k}\;{}\<[15]%
\>[15]{}\Varid{xs}{}\<[23]%
\>[23]{}\mid {}\<[23E]%
\>[26]{}\Varid{k}\doubleequals\Varid{length}\;\Varid{xs}\mathrel{=}[\mskip1.5mu \Varid{xs}\mskip1.5mu]{}\<[E]%
\\
\>[B]{}\Varid{choose}\;(\mathrm{1}\mathbin{+}\Varid{k})\;{}\<[15]%
\>[15]{}(\Varid{x}\mathbin{:}\Varid{xs}){}\<[23]%
\>[23]{}\mathrel{=}{}\<[23E]%
\>[26]{}\Varid{L}\;(\Varid{x}\mathbin{:})\;(\Varid{choose}\;\Varid{k}\;\Varid{xs})\mathbin{{+}\mskip-8mu{+}}\Varid{choose}\;(\mathrm{1}\mathbin{+}\Varid{k})\;\Varid{xs}~~.{}\<[E]%
\ColumnHook
\end{hscode}\resethooks
Its definition follows basic combinatorics: the only way to choose \ensuremath{\mathrm{0}} elements from a list is \ensuremath{[\mskip1.5mu \mskip1.5mu]}; if \ensuremath{\Varid{length}\;\Varid{xs}\mathrel{=}\Varid{k}}, the only way to choose \ensuremath{\Varid{k}} elements is \ensuremath{\Varid{xs}}. Otherwise, to choose \ensuremath{\mathrm{1}\mathbin{+}\Varid{k}} elements from \ensuremath{\Varid{x}\mathbin{:}\Varid{xs}}, one can either keep \ensuremath{\Varid{x}} and choose \ensuremath{\Varid{k}} from \ensuremath{\Varid{xs}}, or discard \ensuremath{\Varid{x}} and choose \ensuremath{\mathrm{1}\mathbin{+}\Varid{k}} elements from \ensuremath{\Varid{xs}}.
For example, \ensuremath{\Varid{choose}\;\mathrm{3}\;\mathtt{\dqoute abcde\dqoute}} yields
\ensuremath{[\mskip1.5mu \mathtt{\dqoute abc\dqoute},\mathtt{\dqoute abd\dqoute},\mathtt{\dqoute abe\dqoute},\mathtt{\dqoute acd\dqoute},\mathtt{\dqoute ace\dqoute},\mathtt{\dqoute ade\dqoute},} \ensuremath{\mathtt{\dqoute bcd\dqoute},\mathtt{\dqoute bce\dqoute},\mathtt{\dqoute bde\dqoute},\mathtt{\dqoute cde\dqoute}\mskip1.5mu]}.

Note that \ensuremath{\Varid{choose}\;\Varid{k}\;\Varid{xs}} is defined only when \ensuremath{\Varid{k}\leq \Varid{length}\;\Varid{xs}}.
Note also that \ensuremath{\Varid{subs}} is a special case of \ensuremath{\Varid{choose}} --- we have \ensuremath{\Varid{subs}\;\Varid{xs}\mathrel{=}\Varid{choose}\;(\Varid{length}\;\Varid{xs}\mathbin{-}\mathrm{1})\;\Varid{xs}}, a property we will need later.

If the levels in Figure~\ref{fig:ch-lattice} were to be represented as lists,
level \ensuremath{\Varid{k}} is given by \ensuremath{\Varid{L}\;\Varid{h}\;(\Varid{choose}\;\Varid{k}\;\Varid{xs})}.
For example, level \ensuremath{\mathrm{2}} in Figure~\ref{fig:ch-lattice} is (string literals are shown in typewriter font; double quotes are omitted to reduce noise in the presentation):
\begin{hscode}\SaveRestoreHook
\column{B}{@{}>{\hspre}l<{\hspost}@{}}%
\column{E}{@{}>{\hspre}l<{\hspost}@{}}%
\>[B]{}\Varid{L}\;\Varid{h}\;(\Varid{choose}\;\mathrm{2}\;\mathtt{abcde})\mathrel{=}[\mskip1.5mu \Varid{h}\;\mathtt{ab},\Varid{h}\;\mathtt{ac},\Varid{h}\;\mathtt{ad},\Varid{h}\;\mathtt{ae},\Varid{h}\;\mathtt{bc},\Varid{h}\;\mathtt{bd},\Varid{h}\;\mathtt{be},\Varid{h}\;\mathtt{cd},\Varid{h}\;\mathtt{ce},\Varid{h}\;\mathtt{de}\mskip1.5mu]~~.{}\<[E]%
\ColumnHook
\end{hscode}\resethooks
To build level \ensuremath{\mathrm{3}} from level \ensuremath{\mathrm{2}},
we wish to have a function \ensuremath{\Varid{upgrade}\mathbin{::}\Conid{L}\;\Conid{Y}\to \Conid{L}\;(\Conid{L}\;\Conid{Y})} that is able to somehow bring together the relevant entries from level \ensuremath{\mathrm{2}}:
\begin{hscode}\SaveRestoreHook
\column{B}{@{}>{\hspre}l<{\hspost}@{}}%
\column{3}{@{}>{\hspre}l<{\hspost}@{}}%
\column{4}{@{}>{\hspre}l<{\hspost}@{}}%
\column{E}{@{}>{\hspre}l<{\hspost}@{}}%
\>[3]{}\Varid{upgrade}\;(\Varid{L}\;\Varid{h}\;(\Varid{choose}\;\mathrm{2}\;\mathtt{abcde}))\mathrel{=}{}\<[E]%
\\
\>[3]{}\hsindent{1}{}\<[4]%
\>[4]{}[\mskip1.5mu [\mskip1.5mu \Varid{h}\;\mathtt{ab},\Varid{h}\;\mathtt{ac},\Varid{h}\;\mathtt{bc}\mskip1.5mu],[\mskip1.5mu \Varid{h}\;\mathtt{ab},\Varid{h}\;\mathtt{ad},\Varid{h}\;\mathtt{bd}\mskip1.5mu],[\mskip1.5mu \Varid{h}\;\mathtt{ab},\Varid{h}\;\mathtt{ae},\Varid{h}\;\mathtt{be}\mskip1.5mu]\mathbin{...}\mskip1.5mu]~~.{}\<[E]%
\ColumnHook
\end{hscode}\resethooks
With \ensuremath{[\mskip1.5mu \Varid{h}\;\mathtt{ab},\Varid{h}\;\mathtt{ac},\Varid{h}\;\mathtt{bc}\mskip1.5mu]} one can compute \ensuremath{\Varid{h}\;\mathtt{abc}}; with
\ensuremath{[\mskip1.5mu \Varid{h}\;\mathtt{ab},\Varid{h}\;\mathtt{ad},\Varid{h}\;\mathtt{bd}\mskip1.5mu]} one can compute \ensuremath{\Varid{h}\;\mathtt{abd}}, and so on.
That is, if we apply \ensuremath{\Varid{L}\;\Varid{g}} to the result of \ensuremath{\Varid{upgrade}} above, we get:
\begin{hscode}\SaveRestoreHook
\column{B}{@{}>{\hspre}l<{\hspost}@{}}%
\column{4}{@{}>{\hspre}l<{\hspost}@{}}%
\column{E}{@{}>{\hspre}l<{\hspost}@{}}%
\>[4]{}[\mskip1.5mu \Varid{h}\;\mathtt{abc},\Varid{h}\;\mathtt{abd},\Varid{h}\;\mathtt{abe},\Varid{h}\;\mathtt{acd}\mathbin{...}\mskip1.5mu]~~,{}\<[E]%
\ColumnHook
\end{hscode}\resethooks
which is level \ensuremath{\mathrm{3}}, or \ensuremath{\Varid{L}\;\Varid{h}\;(\Varid{choose}\;\mathrm{3}\;\mathtt{abcde})}.
The function \ensuremath{\Varid{upgrade}} needs not inspect the values of each element, but rearranges them by position --- it is a natural transformation \ensuremath{\Conid{L}\;\Varid{a}\to \Conid{L}\;(\Conid{L}\;\Varid{a})}.
As far as \ensuremath{\Varid{upgrade}} is concerned, it does not matter whether \ensuremath{\Varid{h}} is applied or not.
Letting \ensuremath{\Varid{h}\mathrel{=}\Varid{id}}, observe that \ensuremath{\Varid{upgrade}\;(\Varid{choose}\;\mathrm{2}\;\mathtt{abcde})\mathrel{=}}
\ensuremath{[\mskip1.5mu [\mskip1.5mu \mathtt{ab},\mathtt{ac},\mathtt{bc}\mskip1.5mu],[\mskip1.5mu \mathtt{ab},\mathtt{ad},\mathtt{bd}\mskip1.5mu]\mathbin{...}\mskip1.5mu]} and
\ensuremath{\Varid{choose}\;\mathrm{3}\;\mathtt{abcde}\mathrel{=}}  \ensuremath{[\mskip1.5mu \mathtt{abc},\mathtt{abd},\mathtt{abe},\mathtt{acd}\mathbin{...}\mskip1.5mu]}
are related by \ensuremath{\Varid{L}\;\Varid{subs}}:
each \ensuremath{\Varid{step}} we perform in the bottom-up algorithm could be \ensuremath{\Varid{L}\;\Varid{g}\circo\Varid{upgrade}}.

Formalising the observations above, we want \ensuremath{\Varid{upgrade}\mathbin{::}\Conid{L}\;\Varid{a}\to \Conid{L}\;(\Conid{L}\;\Varid{a})} to satisfy:
\begin{equation}
\label{eq:up-spec-list}
\begin{split}
&  \ensuremath{(\forall \Varid{xs}\hsforall ,\Varid{k}\mathbin{:}\mathrm{2}\leq \mathrm{1}\mathbin{+}\Varid{k}\leq \Varid{length}\;\Varid{xs}\mathbin{:}}\\
&   \qquad \ensuremath{\Varid{upgrade}\;(\Varid{choose}\;\Varid{k}\;\Varid{xs})\mathrel{=}\Varid{L}\;\Varid{subs}\;(\Varid{choose}\;(\mathrm{1}\mathbin{+}\Varid{k})\;\Varid{xs}))~~.}
\end{split}
\end{equation}
With this property, each \ensuremath{\Varid{step}} we perform in the bottom-up algorithm is \ensuremath{\Varid{L}\;\Varid{g}\circo\Varid{upgrade}}, which converts level \ensuremath{\Varid{k}} to level \ensuremath{\Varid{k}\mathbin{+}\mathrm{1}}:
\begin{hscode}\SaveRestoreHook
\column{B}{@{}>{\hspre}l<{\hspost}@{}}%
\column{5}{@{}>{\hspre}l<{\hspost}@{}}%
\column{9}{@{}>{\hspre}l<{\hspost}@{}}%
\column{E}{@{}>{\hspre}l<{\hspost}@{}}%
\>[5]{}\Varid{L}\;\Varid{g}\circo\Varid{upgrade}\circo\Varid{L}\;\Varid{h}\circo\Varid{choose}\;\Varid{k}{}\<[E]%
\\
\>[B]{}\mathrel{=}{}\<[9]%
\>[9]{}\mbox{\commentbegin  \ensuremath{\Varid{upgrade}} natural  \commentend}{}\<[E]%
\\
\>[B]{}\hsindent{5}{}\<[5]%
\>[5]{}\Varid{L}\;\Varid{g}\circo\Varid{L}\;(\Varid{L}\;\Varid{h})\circo\Varid{upgrade}\circo\Varid{choose}\;\Varid{k}{}\<[E]%
\\
\>[B]{}\mathrel{=}{}\<[9]%
\>[9]{}\mbox{\commentbegin  by \eqref{eq:up-spec-list}, \ensuremath{\Varid{map}}-fusion  \commentend}{}\<[E]%
\\
\>[B]{}\hsindent{5}{}\<[5]%
\>[5]{}\Varid{L}\;(\Varid{g}\circo\Varid{L}\;\Varid{h}\circo\Varid{subs})\circo\Varid{choose}\;(\mathrm{1}\mathbin{+}\Varid{k}){}\<[E]%
\\
\>[B]{}\mathrel{=}{}\<[9]%
\>[9]{}\mbox{\commentbegin  definition of \ensuremath{\Varid{h}}  \commentend}{}\<[E]%
\\
\>[B]{}\hsindent{5}{}\<[5]%
\>[5]{}\Varid{L}\;\Varid{h}\circo\Varid{choose}\;(\mathrm{1}\mathbin{+}\Varid{k})~~.{}\<[E]%
\ColumnHook
\end{hscode}\resethooks

We give some explanation on the constraints on \ensuremath{\Varid{k}} in \eqref{eq:up-spec-list}.
For \ensuremath{\Varid{choose}\;(\mathrm{1}\mathbin{+}\Varid{k})\;\Varid{xs}} on the RHS to be defined, we need \ensuremath{\mathrm{1}\mathbin{+}\Varid{k}\leq \Varid{length}\;\Varid{xs}}.
Meanwhile, no \ensuremath{\Varid{upgrade}} could satisfy \eqref{eq:up-spec-list} when \ensuremath{\Varid{k}\mathrel{=}\mathrm{0}}:
on the LHS, \ensuremath{\Varid{upgrade}} cannot distinguish between \ensuremath{\Varid{choose}\;\mathrm{0}\;\mathtt{ab}} and \ensuremath{\Varid{choose}\;\mathrm{0}\;\mathtt{abc}},
both evaluating to \ensuremath{[\mskip1.5mu [\mskip1.5mu \mskip1.5mu]\mskip1.5mu]}, while on the RHS \ensuremath{\Varid{choose}\;\mathrm{1}\;\mathtt{ab}} and \ensuremath{\Varid{choose}\;\mathrm{1}\;\mathtt{abc}} have different shapes.
Therefore we only demand \eqref{eq:up-spec-list} to hold when \ensuremath{\mathrm{1}\leq \Varid{k}},
which is sufficient because we only apply \ensuremath{\Varid{upgrade}} to level \ensuremath{\mathrm{1}} and above.
Together, the constraint is \ensuremath{\mathrm{2}\leq \mathrm{1}\mathbin{+}\Varid{k}\leq \Varid{length}\;\Varid{xs}} --- \ensuremath{\Varid{xs}} should have at least \ensuremath{\mathrm{2}} elements.

Can we construct such an \ensuremath{\Varid{upgrade}}?

\section{Building Levels Represented By Trees}

We may proceed with \eqref{eq:up-spec-list} and construct \ensuremath{\Varid{upgrade}}.
We will soon meet a small obstacle: in an inductive case
\ensuremath{\Varid{upgrade}} will receive a list computed by \ensuremath{\Varid{choose}\;(\mathrm{1}\mathbin{+}\Varid{k})\;(\Varid{x}\mathbin{:}\Varid{xs})} that
needs to be split into \ensuremath{\Varid{L}\;(\Varid{x}\mathbin{:})\;(\Varid{choose}\;\Varid{k}\;\Varid{xs})} and \ensuremath{\Varid{choose}\;(\mathrm{1}\mathbin{+}\Varid{k})\;\Varid{xs}}.
This can be done, but it is rather tedious.
This is a hint that some useful information has been lost when we represent levels by lists.
To make the job of \ensuremath{\Varid{upgrade}} easier, we switch to a more informative data structure.

\subsection{Binomial Trees}

Instead of lists, we define the following tip-valued binary tree:
\begin{hscode}\SaveRestoreHook
\column{B}{@{}>{\hspre}l<{\hspost}@{}}%
\column{E}{@{}>{\hspre}l<{\hspost}@{}}%
\>[B]{}\mathbf{data}\;\Conid{B}\;\Varid{a}\mathrel{=}\Conid{T}\;\Varid{a}\mid \Conid{N}\;(\Conid{B}\;\Varid{a})\;(\Conid{B}\;\Varid{a})~~.{}\<[E]%
\ColumnHook
\end{hscode}\resethooks
We assume that \ensuremath{\Conid{B}} is equipped with two functions derived from its definition:
\begin{hscode}\SaveRestoreHook
\column{B}{@{}>{\hspre}l<{\hspost}@{}}%
\column{8}{@{}>{\hspre}l<{\hspost}@{}}%
\column{E}{@{}>{\hspre}l<{\hspost}@{}}%
\>[B]{}\Varid{mapB}{}\<[8]%
\>[8]{}\mathbin{::}(\Varid{a}\to \Varid{b})\to \Conid{B}\;\Varid{a}\to \Conid{B}\;\Varid{b}~~,{}\<[E]%
\\
\>[B]{}\Varid{zipBW}{}\<[8]%
\>[8]{}\mathbin{::}(\Varid{a}\to \Varid{b}\to \Varid{c})\to \Conid{B}\;\Varid{a}\to \Conid{B}\;\Varid{b}\to \Conid{B}\;\Varid{c}~~.{}\<[E]%
\ColumnHook
\end{hscode}\resethooks
The function \ensuremath{\Varid{mapB}\;\Varid{f}} applies \ensuremath{\Varid{f}} to every tip of the given tree.
Like that with lists, we also write \ensuremath{\Varid{mapB}} as $\Varid{B}$.
Given two trees \ensuremath{\Varid{t}} and \ensuremath{\Varid{u}} having the same shape,
\ensuremath{\Varid{zipBW}\;\Varid{f}\;\Varid{t}\;\Varid{u}} ``zips'' the trees together, applying \ensuremath{\Varid{f}} to values on the tips.
If \ensuremath{\Varid{t}} and \ensuremath{\Varid{u}} have different shapes, \ensuremath{\Varid{zipBW}\;\Varid{f}\;\Varid{t}\;\Varid{u}} is undefined.

Having \ensuremath{\Conid{B}} allows us to define an alternative to \ensuremath{\Varid{choose}}:
\begin{hscode}\SaveRestoreHook
\column{B}{@{}>{\hspre}l<{\hspost}@{}}%
\column{11}{@{}>{\hspre}l<{\hspost}@{}}%
\column{19}{@{}>{\hspre}c<{\hspost}@{}}%
\column{19E}{@{}l@{}}%
\column{22}{@{}>{\hspre}l<{\hspost}@{}}%
\column{E}{@{}>{\hspre}l<{\hspost}@{}}%
\>[B]{}\Varid{ch}\mathbin{::}\Conid{Nat}\to \Conid{L}\;\Varid{a}\to \Conid{B}\;(\Conid{L}\;\Varid{a}){}\<[E]%
\\
\>[B]{}\Varid{ch}\;\mathrm{0}\;{}\<[11]%
\>[11]{}\anonymous {}\<[19]%
\>[19]{}\mathrel{=}{}\<[19E]%
\>[22]{}\Conid{T}\;[\mskip1.5mu \mskip1.5mu]{}\<[E]%
\\
\>[B]{}\Varid{ch}\;\Varid{k}\;{}\<[11]%
\>[11]{}\Varid{xs}{}\<[19]%
\>[19]{}\mid {}\<[19E]%
\>[22]{}\Varid{k}\doubleequals\Varid{length}\;\Varid{xs}\mathrel{=}\Conid{T}\;\Varid{xs}{}\<[E]%
\\
\>[B]{}\Varid{ch}\;(\mathrm{1}\mathbin{+}\Varid{k})\;{}\<[11]%
\>[11]{}(\Varid{x}\mathbin{:}\Varid{xs}){}\<[19]%
\>[19]{}\mathrel{=}{}\<[19E]%
\>[22]{}\Conid{N}\;(\Varid{B}\;(\Varid{x}\mathbin{:})\;(\Varid{ch}\;\Varid{k}\;\Varid{xs}))\;(\Varid{ch}\;(\mathrm{1}\mathbin{+}\Varid{k})\;\Varid{xs})~~.{}\<[E]%
\ColumnHook
\end{hscode}\resethooks
The function \ensuremath{\Varid{ch}} resembles \ensuremath{\Varid{choose}}.
In the first two clauses, \ensuremath{\Conid{T}} corresponds to a singleton list.
In the last clause, \ensuremath{\Varid{ch}} is like \ensuremath{\Varid{choose}} but, instead of appending the results of the two recursive calls, we store the results in the two branches of the binary tree, thus recording how the choices were made:
if \ensuremath{\Varid{ch}\;\anonymous \;(\Varid{x}\mathbin{:}\Varid{xs})\mathrel{=}\Conid{N}\;\Varid{t}\;\Varid{u}}, the subtree \ensuremath{\Varid{t}} contains all the tips with \ensuremath{\Varid{x}} chosen, while \ensuremath{\Varid{u}} contains all the tips with \ensuremath{\Varid{x}} discarded.

\begin{figure}
\centering
\begin{subfigure}[b]{0.25\textwidth}
  \centering
  \includegraphics[width=0.6\textwidth]{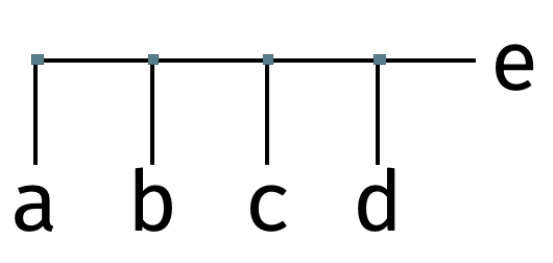}
  \caption{\ensuremath{\Varid{ch}\;\mathrm{1}\;\mathtt{abcde}}.}
  \label{fig:ch-1-5}
\end{subfigure}
\qquad
\begin{subfigure}[b]{0.45\textwidth}
  \centering
  \includegraphics[width=0.8\textwidth]{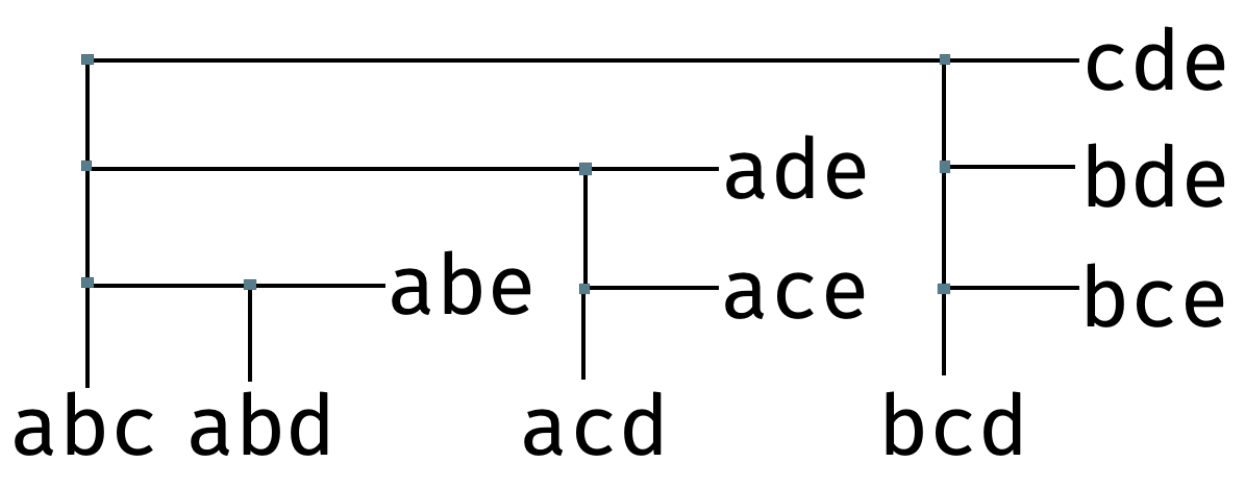}
  \caption{\ensuremath{\Varid{ch}\;\mathrm{3}\;\mathtt{abcde}}.}
  \label{fig:ch-3-5}
\end{subfigure}
\caption{Results of \ensuremath{\Varid{ch}}.}
\label{fig:ch-examples}
\end{figure}

The counterpart of \ensuremath{\Varid{upgrade}} on trees, which we will call \ensuremath{\Varid{up}}, is a natural transformation of type \ensuremath{\Conid{B}\;\Varid{a}\to \Conid{B}\;(\Conid{L}\;\Varid{a})}, satisfying the following property:
\begin{equation}
\label{eq:up-spec-B}
\begin{split}
&  \ensuremath{(\forall \Varid{xs}\hsforall ,\Varid{k}\mathbin{:}\mathrm{2}\leq \mathrm{1}\mathbin{+}\Varid{k}\leq \Varid{length}\;\Varid{xs}} \,: \\
&   \qquad \ensuremath{\Varid{up}\;(\Varid{ch}\;\Varid{k}\;\Varid{xs})\mathrel{=}\Varid{B}\;\Varid{subs}\;(\Varid{ch}\;(\mathrm{1}\mathbin{+}\Varid{k})\;\Varid{xs}))~~.}
\end{split}
\end{equation}

Now we are ready to derive \ensuremath{\Varid{up}}.

\subsection{The Derivation}

The derivation proceeds by trying to construct a proof of \eqref{eq:up-spec-B} and, when stuck, pausing to think about how \ensuremath{\Varid{up}} should be defined to allow the proof to go through.
That is, the definition of \ensuremath{\Varid{up}} and a proof that it satisfies \eqref{eq:up-spec-B} are developed hand-in-hand.

The proof, if constructed, will be an induction on \ensuremath{\Varid{xs}}.
The case analysis follows the shape of \ensuremath{\Varid{ch}\;(\mathrm{1}\mathbin{+}\Varid{k})\;\Varid{xs}} (on the RHS of \eqref{eq:up-spec-B}).
Therefore, there is a base case, a case when \ensuremath{\Varid{xs}} is non-empty and \ensuremath{\mathrm{1}\mathbin{+}\Varid{k}\mathrel{=}\Varid{length}\;\Varid{xs}}, and a case when \ensuremath{\mathrm{1}\mathbin{+}\Varid{k}\mathbin{<}\Varid{length}\;\Varid{xs}}.
However, since the constraints demand that \ensuremath{\Varid{xs}} has at least two elements, the base case will be lists of length \ensuremath{\mathrm{2}}, and in the inductive cases the length of the list will be at least \ensuremath{\mathrm{3}}.

\paragraph*{Case 1.~~} \ensuremath{\Varid{xs}\mathbin{:=}[\mskip1.5mu \Varid{y},\Varid{z}\mskip1.5mu]}.\\
The constraints force \ensuremath{\Varid{k}} to be \ensuremath{\mathrm{1}}.
We simplify the RHS of \eqref{eq:up-spec-B}:
\begin{hscode}\SaveRestoreHook
\column{B}{@{}>{\hspre}l<{\hspost}@{}}%
\column{3}{@{}>{\hspre}c<{\hspost}@{}}%
\column{3E}{@{}l@{}}%
\column{5}{@{}>{\hspre}l<{\hspost}@{}}%
\column{8}{@{}>{\hspre}l<{\hspost}@{}}%
\column{E}{@{}>{\hspre}l<{\hspost}@{}}%
\>[5]{}\Varid{B}\;\Varid{subs}\;(\Varid{ch}\;\mathrm{2}\;[\mskip1.5mu \Varid{y},\Varid{z}\mskip1.5mu]){}\<[E]%
\\
\>[3]{}\mathrel{=}{}\<[3E]%
\>[8]{}\mbox{\commentbegin  def. of \ensuremath{\Varid{ch}}  \commentend}{}\<[E]%
\\
\>[3]{}\hsindent{2}{}\<[5]%
\>[5]{}\Varid{B}\;\Varid{subs}\;(\Conid{T}\;[\mskip1.5mu \Varid{y},\Varid{z}\mskip1.5mu]){}\<[E]%
\\
\>[3]{}\mathrel{=}{}\<[3E]%
\>[8]{}\mbox{\commentbegin  def. of \ensuremath{\Varid{B}} and \ensuremath{\Varid{subs}}  \commentend}{}\<[E]%
\\
\>[3]{}\hsindent{2}{}\<[5]%
\>[5]{}\Conid{T}\;[\mskip1.5mu [\mskip1.5mu \Varid{y}\mskip1.5mu],[\mskip1.5mu \Varid{z}\mskip1.5mu]\mskip1.5mu]~~.{}\<[E]%
\ColumnHook
\end{hscode}\resethooks
Now consider the LHS:
\begin{hscode}\SaveRestoreHook
\column{B}{@{}>{\hspre}l<{\hspost}@{}}%
\column{3}{@{}>{\hspre}c<{\hspost}@{}}%
\column{3E}{@{}l@{}}%
\column{5}{@{}>{\hspre}l<{\hspost}@{}}%
\column{8}{@{}>{\hspre}l<{\hspost}@{}}%
\column{E}{@{}>{\hspre}l<{\hspost}@{}}%
\>[5]{}\Varid{up}\;(\Varid{ch}\;\mathrm{1}\;[\mskip1.5mu \Varid{y},\Varid{z}\mskip1.5mu]){}\<[E]%
\\
\>[3]{}\mathrel{=}{}\<[3E]%
\>[8]{}\mbox{\commentbegin  def. of \ensuremath{\Varid{ch}}  \commentend}{}\<[E]%
\\
\>[3]{}\hsindent{2}{}\<[5]%
\>[5]{}\Varid{up}\;(\Conid{N}\;(\Conid{T}\;[\mskip1.5mu \Varid{y}\mskip1.5mu])\;(\Conid{T}\;[\mskip1.5mu \Varid{z}\mskip1.5mu]))~~.{}\<[E]%
\ColumnHook
\end{hscode}\resethooks
The two sides can be made equal if we let \ensuremath{\Varid{up}\;(\Conid{N}\;(\Conid{T}\;\Varid{p})\;(\Conid{T}\;\Varid{q}))\mathrel{=}\Conid{T}\;[\mskip1.5mu \Varid{p},\Varid{q}\mskip1.5mu]}.

\paragraph*{\bf Case 2.~~} \ensuremath{\Varid{xs}\mathbin{:=}\Varid{x}\mathbin{:}\Varid{xs}} where \ensuremath{\Varid{length}\;\Varid{xs}\geq \mathrm{2}}, and \ensuremath{\mathrm{1}\mathbin{+}\Varid{k}\mathrel{=}\Varid{length}\;(\Varid{x}\mathbin{:}\Varid{xs})}.\\
We leave details of this case to the readers as an exercise, since we would prefer giving more attention to the next case.
For this case we will construct
\begin{hscode}\SaveRestoreHook
\column{B}{@{}>{\hspre}l<{\hspost}@{}}%
\column{E}{@{}>{\hspre}l<{\hspost}@{}}%
\>[B]{}\Varid{up}\;(\Conid{N}\;\Varid{t}\;(\Conid{T}\;\Varid{q}))\mathrel{=}\Conid{T}\;(\Varid{unT}\;(\Varid{up}\;\Varid{t})\mathbin{{+}\mskip-8mu{+}}[\mskip1.5mu \Varid{q}\mskip1.5mu])~~.{}\<[E]%
\ColumnHook
\end{hscode}\resethooks
In this case, \ensuremath{\Varid{up}\;\Varid{t}} always returns a \ensuremath{\Conid{T}}.
The function \ensuremath{\Varid{unT}\;(\Conid{T}\;\Varid{p})\mathrel{=}\Varid{p}} removes the constructor and exposes the list it contains.
While the correctness of this case is established by the constructed proof, a complementary explanation why \ensuremath{\Varid{up}\;\Varid{t}} always returns a singleton tree and thus \ensuremath{\Varid{unT}} always succeeds is given in Section~\ref{sec:deptypes}.

\paragraph*{\bf Case 3.~~} \ensuremath{\Varid{xs}\mathbin{:=}\Varid{x}\mathbin{:}\Varid{xs}}, \ensuremath{\Varid{k}\mathbin{:=}\mathrm{1}\mathbin{+}\Varid{k}}, where \ensuremath{\Varid{length}\;\Varid{xs}\geq \mathrm{2}}, and \ensuremath{\mathrm{1}\mathbin{+}(\mathrm{1}\mathbin{+}\Varid{k})\mathbin{<}\Varid{length}\;(\Varid{x}\mathbin{:}\Varid{xs})}.
\\
The constraints become \ensuremath{\mathrm{2}\leq \mathrm{2}\mathbin{+}\Varid{k}\mathbin{<}\Varid{length}\;(\Varid{x}\mathbin{:}\Varid{xs})}.
Again we start with the RHS, and try to reach the LHS:
\addtolength\jot{-2.2pt}
\begin{align*}
  & \ensuremath{\Varid{B}\;\Varid{subs}\;(\Varid{ch}\;(\mathrm{2}\mathbin{+}\Varid{k})\;(\Varid{x}\mathbin{:}\Varid{xs}))} \\
=~& \mbox{\color{SeaGreen}\quad\{ def. of \ensuremath{\Varid{ch}}, since \ensuremath{\mathrm{2}\mathbin{+}\Varid{k}\mathbin{<}\Varid{length}\;(\Varid{x}\mathbin{:}\Varid{xs})} \}} \\
  & \ensuremath{\Varid{B}\;\Varid{subs}\;(\Conid{N}\;(\Varid{B}\;(\Varid{x}\mathbin{:})\;(\Varid{ch}\;(\mathrm{1}\mathbin{+}\Varid{k})\;\Varid{xs}))\;(\Varid{ch}\;(\mathrm{2}\mathbin{+}\Varid{k})\;\Varid{xs}))} \\
=~& \mbox{\color{SeaGreen}\quad\{ def. of \ensuremath{\Varid{B}} \}}\\
  & \ensuremath{\Conid{N}\;(\Varid{B}\;(\Varid{subs}\circo(\Varid{x}\mathbin{:}))\;(\Varid{ch}\;(\mathrm{1}\mathbin{+}\Varid{k})\;\Varid{xs}))\;(\Varid{B}\;\Varid{subs}\;(\Varid{ch}\;(\mathrm{2}\mathbin{+}\Varid{k})\;\Varid{xs}))}\\
=~& \mbox{\color{SeaGreen}\quad\{ induction \}}\\
  & \ensuremath{\Conid{N}\;(\Varid{B}\;(\Varid{subs}\circo(\Varid{x}\mathbin{:}))\;(\Varid{ch}\;(\mathrm{1}\mathbin{+}\Varid{k})\;\Varid{xs}))\;(\Varid{up}\;(\Varid{ch}\;(\mathrm{1}\mathbin{+}\Varid{k})\;\Varid{xs}))}
    \mbox{~~.} \numberthis \label{eq:up3R}
\end{align*}
\addtolength\jot{2.2pt}%
Note that the induction step is valid because we are performing induction on \ensuremath{\Varid{xs}}, and thus \ensuremath{\Varid{k}} in \eqref{eq:up-spec-B} is universally quantified.
We now look at the LHS:
\addtolength\jot{-2.2pt}
\begin{align*}
  & \ensuremath{\Varid{up}\;(\Varid{ch}\;(\mathrm{1}\mathbin{+}\Varid{k})\;(\Varid{x}\mathbin{:}\Varid{xs}))} \\
=~& \mbox{\color{SeaGreen}\quad\{ def. of \ensuremath{\Varid{ch}}, since \ensuremath{\mathrm{1}\mathbin{+}\Varid{k}\mathbin{<}\Varid{length}\;(\Varid{x}\mathbin{:}\Varid{xs})} \}}\\
  & \ensuremath{\Varid{up}\;(\Conid{N}\;(\Varid{B}\;(\Varid{x}\mathbin{:})\;(\Varid{ch}\;\Varid{k}\;\Varid{xs}))\;(\Varid{ch}\;(\mathrm{1}\mathbin{+}\Varid{k})\;\Varid{xs}))}
     \mbox{~~.} \numberthis \label{eq:up3L}
\end{align*}
\addtolength\jot{2.2pt}%
Expressions \eqref{eq:up3R} and \eqref{eq:up3L} can be unified if we define
\begin{hscode}\SaveRestoreHook
\column{B}{@{}>{\hspre}l<{\hspost}@{}}%
\column{6}{@{}>{\hspre}l<{\hspost}@{}}%
\column{E}{@{}>{\hspre}l<{\hspost}@{}}%
\>[6]{}\Varid{up}\;(\Conid{N}\;\Varid{t}\;\Varid{u})\mathrel{=}\Conid{N}\mathbin{???}(\Varid{up}\;\Varid{u})~~.{}\<[E]%
\ColumnHook
\end{hscode}\resethooks
The missing part \ensuremath{\mathbin{???}} shall be an expression that is allowed to use only the two subtrees \ensuremath{\Varid{t}} and \ensuremath{\Varid{u}} that \ensuremath{\Varid{up}} receives.
Given \ensuremath{\Varid{t}\mathrel{=}\Varid{B}\;(\Varid{x}\mathbin{:})\;(\Varid{ch}\;\Varid{k}\;\Varid{xs})} and \ensuremath{\Varid{u}\mathrel{=}\Varid{ch}\;(\mathrm{1}\mathbin{+}\Varid{k})\;\Varid{xs}} (from \eqref{eq:up3L})
this expression shall evaluate to the subexpression in \eqref{eq:up3R}:
\begin{hscode}\SaveRestoreHook
\column{B}{@{}>{\hspre}l<{\hspost}@{}}%
\column{5}{@{}>{\hspre}l<{\hspost}@{}}%
\column{E}{@{}>{\hspre}l<{\hspost}@{}}%
\>[5]{}\Varid{B}\;(\Varid{subs}\circo(\Varid{x}\mathbin{:}))\;(\Varid{ch}\;(\mathrm{1}\mathbin{+}\Varid{k})\;\Varid{xs})~~.{}\<[E]%
\ColumnHook
\end{hscode}\resethooks

It may appear that, now that \ensuremath{\Varid{up}} already has \ensuremath{\Varid{u}\mathrel{=}\Varid{ch}\;(\mathrm{1}\mathbin{+}\Varid{k})\;\Varid{xs}}, the \ensuremath{\mathbin{???}} may simply be \ensuremath{\Varid{B}\;(\Varid{sub}\circo(\Varid{x}\mathbin{:}))\;\Varid{u}}. The problem is that the \ensuremath{\Varid{up}} does not know what \ensuremath{\Varid{x}} is --- unless \ensuremath{\Varid{k}\mathrel{=}\mathrm{0}}.

\paragraph*{Case 3.1.~~} \ensuremath{\Varid{k}\mathrel{=}\mathrm{0}}.
We can recover \ensuremath{\Varid{x}} from \ensuremath{\Varid{B}\;(\Varid{x}\mathbin{:})\;(\Varid{ch}\;\mathrm{0}\;\Varid{xs})} if \ensuremath{\Varid{k}} happens to be \ensuremath{\mathrm{0}} because:
\begin{hscode}\SaveRestoreHook
\column{B}{@{}>{\hspre}l<{\hspost}@{}}%
\column{6}{@{}>{\hspre}c<{\hspost}@{}}%
\column{6E}{@{}l@{}}%
\column{11}{@{}>{\hspre}l<{\hspost}@{}}%
\column{E}{@{}>{\hspre}l<{\hspost}@{}}%
\>[11]{}\Varid{B}\;(\Varid{x}\mathbin{:})\;(\Varid{ch}\;\mathrm{0}\;\Varid{xs}){}\<[E]%
\\
\>[6]{}\mathrel{=}{}\<[6E]%
\>[11]{}\Varid{B}\;(\Varid{x}\mathbin{:})\;(\Conid{T}\;[\mskip1.5mu \mskip1.5mu]){}\<[E]%
\\
\>[6]{}\mathrel{=}{}\<[6E]%
\>[11]{}\Conid{T}\;[\mskip1.5mu \Varid{x}\mskip1.5mu]~~.{}\<[E]%
\ColumnHook
\end{hscode}\resethooks
That is, the left subtree \ensuremath{\Varid{up}} receives must have the form \ensuremath{\Conid{T}\;[\mskip1.5mu \Varid{x}\mskip1.5mu]},
from which can retrieve \ensuremath{\Varid{x}} and apply \ensuremath{\Varid{B}\;(\Varid{sub}\circo(\Varid{x}\mathbin{:}))} to the other subtree.
We can furthermore simplify \ensuremath{\Varid{B}\;(\Varid{sub}\circo(\Varid{x}\mathbin{:}))\;(\Varid{ch}\;(\mathrm{1}\mathbin{+}\mathrm{0})\;\Varid{xs})} a bit:
\begin{hscode}\SaveRestoreHook
\column{B}{@{}>{\hspre}l<{\hspost}@{}}%
\column{6}{@{}>{\hspre}c<{\hspost}@{}}%
\column{6E}{@{}l@{}}%
\column{11}{@{}>{\hspre}l<{\hspost}@{}}%
\column{E}{@{}>{\hspre}l<{\hspost}@{}}%
\>[11]{}\Varid{B}\;(\Varid{subs}\circo(\Varid{x}\mathbin{:}))\;(\Varid{ch}\;(\mathrm{1}\mathbin{+}\mathrm{0})\;\Varid{xs}){}\<[E]%
\\
\>[6]{}\mathrel{=}{}\<[6E]%
\>[11]{}\Varid{B}\;(\lambda \Varid{q}\to [\mskip1.5mu [\mskip1.5mu \Varid{x}\mskip1.5mu],\Varid{q}\mskip1.5mu])\;(\Varid{ch}\;\mathrm{1}\;\Varid{xs})~~.{}\<[E]%
\ColumnHook
\end{hscode}\resethooks
The equality above holds because every tip in \ensuremath{\Varid{ch}\;\mathrm{1}\;\Varid{xs}} contains singleton lists and, for a singleton list \ensuremath{[\mskip1.5mu \Varid{z}\mskip1.5mu]}, we have \ensuremath{\Varid{subs}\;(\Varid{x}\mathbin{:}[\mskip1.5mu \Varid{z}\mskip1.5mu])\mathrel{=}[\mskip1.5mu [\mskip1.5mu \Varid{x}\mskip1.5mu],[\mskip1.5mu \Varid{z}\mskip1.5mu]\mskip1.5mu]}.
In summary, we have established
\begin{hscode}\SaveRestoreHook
\column{B}{@{}>{\hspre}l<{\hspost}@{}}%
\column{7}{@{}>{\hspre}l<{\hspost}@{}}%
\column{E}{@{}>{\hspre}l<{\hspost}@{}}%
\>[7]{}\Varid{up}\;(\Conid{N}\;(\Conid{T}\;\Varid{p})\;\Varid{u})\mathrel{=}\Conid{N}\;(\Varid{B}\;(\lambda \Varid{q}\to [\mskip1.5mu \Varid{p},\Varid{q}\mskip1.5mu])\;\Varid{u})\;(\Varid{up}\;\Varid{u})~~.{}\<[E]%
\ColumnHook
\end{hscode}\resethooks

\paragraph*{Case 3.2.~~} \ensuremath{\mathrm{0}\mathbin{<}\Varid{k}} (and \ensuremath{\Varid{k}\mathbin{<}\Varid{length}\;\Varid{xs}\mathbin{-}\mathrm{1}}).
In this more general case, we have to construct \ensuremath{\Varid{B}\;(\Varid{subs}\circo(\Varid{x}\mathbin{:}))\;(\Varid{ch}\;(\mathrm{1}\mathbin{+}\Varid{k})\;\Varid{xs})} out of the two subtrees, \ensuremath{\Varid{B}\;(\Varid{x}\mathbin{:})\;(\Varid{ch}\;\Varid{k}\;\Varid{xs})} and \ensuremath{\Varid{ch}\;(\mathrm{1}\mathbin{+}\Varid{k})\;\Varid{xs}}, without knowing what \ensuremath{\Varid{x}} is.

Starting calculation from \ensuremath{\Varid{B}\;(\Varid{subs}\circo(\Varid{x}\mathbin{:}))\;(\Varid{ch}\;(\mathrm{1}\mathbin{+}\Varid{k})\;\Varid{xs})},
we expect to use induction somewhere, therefore a possible strategy is to move \ensuremath{\Varid{B}\;\Varid{subs}} rightwards, closer to \ensuremath{\Varid{ch}}, in order to apply \eqref{eq:up-spec-B}.
Let us consider how to compute \ensuremath{\Varid{B}\;(\Varid{subs}\circo(\Varid{x}\mathbin{:}))\;\Varid{u}} for a general \ensuremath{\Varid{u}}, and try to move \ensuremath{\Varid{B}\;\Varid{subs}} closer to \ensuremath{\Varid{u}}.
Note that
\begin{itemize}
\item by definition, \ensuremath{\Varid{sub}\;(\Varid{x}\mathbin{:}\Varid{xs})\mathrel{=}\Varid{L}\;(\Varid{x}\mathbin{:})\;(\Varid{sub}\;\Varid{xs})\mathbin{{+}\mskip-8mu{+}}[\mskip1.5mu \Varid{xs}\mskip1.5mu]}.
\item Given a tree \ensuremath{\Varid{u}} and functions \ensuremath{\Varid{f}}, \ensuremath{\Varid{g}}, and \ensuremath{\Varid{h}},
by naturality of \ensuremath{\Varid{zipBW}} we have:
\begin{equation}
\label{eq:map-zipBW}
 \ensuremath{\Varid{B}\;(\lambda \Varid{z}\to \Varid{f}\;(\Varid{g}\;\Varid{z})\;(\Varid{h}\;\Varid{z}))\;\Varid{u}\mathrel{=}\Varid{zipBW}\;\Varid{f}\;(\Varid{B}\;\Varid{g}\;\Varid{u})\;(\Varid{B}\;\Varid{h}\;\Varid{u})~~.}
\end{equation}
\item Therefore, letting \ensuremath{\Varid{g}\mathrel{=}\Varid{L}\;(\Varid{x}\mathbin{:})\circo\Varid{subs}}, \ensuremath{\Varid{h}\mathrel{=}\Varid{id}}, and \ensuremath{\Varid{f}\mathrel{=}\Varid{snoc}} in \eqref{eq:map-zipBW}, where \ensuremath{\Varid{snoc}\;\Varid{ys}\;\Varid{z}\mathrel{=}\Varid{ys}\mathbin{{+}\mskip-8mu{+}}[\mskip1.5mu \Varid{z}\mskip1.5mu]}, we have:
\begin{equation}
\label{eq:map-sub-zipBW}
  \ensuremath{\Varid{B}\;(\Varid{subs}\circo(\Varid{x}\mathbin{:}))\;\Varid{u}\mathrel{=}\Varid{zipBW}\;\Varid{snoc}\;(\Varid{B}\;(\Varid{L}\;(\Varid{x}\mathbin{:})\circo\Varid{subs})\;\Varid{u})\;\Varid{u}~~.}
\end{equation}
\end{itemize}

We calculate:
\begin{hscode}\SaveRestoreHook
\column{B}{@{}>{\hspre}l<{\hspost}@{}}%
\column{5}{@{}>{\hspre}c<{\hspost}@{}}%
\column{5E}{@{}l@{}}%
\column{9}{@{}>{\hspre}l<{\hspost}@{}}%
\column{11}{@{}>{\hspre}l<{\hspost}@{}}%
\column{21}{@{}>{\hspre}l<{\hspost}@{}}%
\column{E}{@{}>{\hspre}l<{\hspost}@{}}%
\>[9]{}\Varid{B}\;(\Varid{subs}\circo(\Varid{x}\mathbin{:}))\;(\Varid{ch}\;(\mathrm{1}\mathbin{+}\Varid{k})\;\Varid{xs}){}\<[E]%
\\
\>[5]{}\mathrel{=}{}\<[5E]%
\>[11]{}\mbox{\commentbegin  by \eqref{eq:map-sub-zipBW}  \commentend}{}\<[E]%
\\
\>[5]{}\hsindent{4}{}\<[9]%
\>[9]{}\Varid{zipBW}\;\Varid{snoc}\;{}\<[21]%
\>[21]{}(\Varid{B}\;(\Varid{L}\;(\Varid{x}\mathbin{:})\circo\Varid{subs})\;(\Varid{ch}\;(\mathrm{1}\mathbin{+}\Varid{k})\;\Varid{xs}))\;(\Varid{ch}\;(\mathrm{1}\mathbin{+}\Varid{k})\;\Varid{xs}){}\<[E]%
\\
\>[5]{}\mathrel{=}{}\<[5E]%
\>[11]{}\mbox{\commentbegin  induction  \commentend}{}\<[E]%
\\
\>[5]{}\hsindent{4}{}\<[9]%
\>[9]{}\Varid{zipBW}\;\Varid{snoc}\;{}\<[21]%
\>[21]{}(\Varid{B}\;(\Varid{L}\;(\Varid{x}\mathbin{:}))\circo\Varid{up}\circo\Varid{ch}\;\Varid{k}\myapply\Varid{xs})\;(\Varid{ch}\;(\mathrm{1}\mathbin{+}\Varid{k})\;\Varid{xs}){}\<[E]%
\\
\>[5]{}\mathrel{=}{}\<[5E]%
\>[11]{}\mbox{\commentbegin  \ensuremath{\Varid{up}} natural  \commentend}{}\<[E]%
\\
\>[5]{}\hsindent{4}{}\<[9]%
\>[9]{}\Varid{zipBW}\;\Varid{snoc}\;(\Varid{up}\circo\Varid{B}\;(\Varid{x}\mathbin{:})\circo\Varid{ch}\;\Varid{k}\myapply\Varid{xs})\;(\Varid{ch}\;(\mathrm{1}\mathbin{+}\Varid{k})\;\Varid{xs})~~.{}\<[E]%
\ColumnHook
\end{hscode}\resethooks
Recall that our aim is to find a suitable definition of \ensuremath{\Varid{up}} such that \eqref{eq:up3R} equals \eqref{eq:up3L} .
The calculation shows that we may let
\begin{hscode}\SaveRestoreHook
\column{B}{@{}>{\hspre}l<{\hspost}@{}}%
\column{8}{@{}>{\hspre}l<{\hspost}@{}}%
\column{E}{@{}>{\hspre}l<{\hspost}@{}}%
\>[8]{}\Varid{up}\;(\Conid{N}\;\Varid{t}\;\Varid{u})\mathrel{=}\Conid{N}\;(\Varid{zipBW}\;\Varid{snoc}\;(\Varid{up}\;\Varid{t})\;\Varid{u})\;(\Varid{up}\;\Varid{u})~~.{}\<[E]%
\ColumnHook
\end{hscode}\resethooks

In summary, we have constructed:
\begin{hscode}\SaveRestoreHook
\column{B}{@{}>{\hspre}l<{\hspost}@{}}%
\column{14}{@{}>{\hspre}l<{\hspost}@{}}%
\column{21}{@{}>{\hspre}l<{\hspost}@{}}%
\column{E}{@{}>{\hspre}l<{\hspost}@{}}%
\>[B]{}\Varid{up}\mathbin{::}\Conid{B}\;\Varid{a}\to \Conid{B}\;(\Conid{L}\;\Varid{a}){}\<[E]%
\\
\>[B]{}\Varid{up}\;(\Conid{N}\;(\Conid{T}\;\Varid{p})\;{}\<[14]%
\>[14]{}(\Conid{T}\;\Varid{q}){}\<[21]%
\>[21]{})\mathrel{=}\Conid{T}\;[\mskip1.5mu \Varid{p},\Varid{q}\mskip1.5mu]{}\<[E]%
\\
\>[B]{}\Varid{up}\;(\Conid{N}\;\Varid{t}\;{}\<[14]%
\>[14]{}(\Conid{T}\;\Varid{q}){}\<[21]%
\>[21]{})\mathrel{=}\Conid{T}\;(\Varid{unT}\;(\Varid{up}\;\Varid{t})\mathbin{{+}\mskip-8mu{+}}[\mskip1.5mu \Varid{q}\mskip1.5mu]){}\<[E]%
\\
\>[B]{}\Varid{up}\;(\Conid{N}\;(\Conid{T}\;\Varid{p})\;{}\<[14]%
\>[14]{}\Varid{u}{}\<[21]%
\>[21]{})\mathrel{=}\Conid{N}\;(\Varid{B}\;(\lambda \Varid{q}\to [\mskip1.5mu \Varid{p},\Varid{q}\mskip1.5mu])\;\Varid{u})\;(\Varid{up}\;\Varid{u}){}\<[E]%
\\
\>[B]{}\Varid{up}\;(\Conid{N}\;\Varid{t}\;{}\<[14]%
\>[14]{}\Varid{u}{}\<[21]%
\>[21]{})\mathrel{=}\Conid{N}\;(\Varid{zipBW}\;\Varid{snoc}\;(\Varid{up}\;\Varid{t})\;\Varid{u})\;(\Varid{up}\;\Varid{u})~~.{}\<[E]%
\ColumnHook
\end{hscode}\resethooks
Using \ensuremath{(\mathbin{{+}\mskip-8mu{+}})} and \ensuremath{\Varid{snoc}} may look inefficient, but had we specified \ensuremath{\Varid{choose}} slightly differently, the \ensuremath{\Varid{up}} we derive would use \ensuremath{(\mathbin{:})} instead.
Again, we defined \ensuremath{\Varid{choose}} this way merely to generate sublists in an intuitive order.

\begin{figure}[h]
\centering
\includegraphics[width=0.9\textwidth]{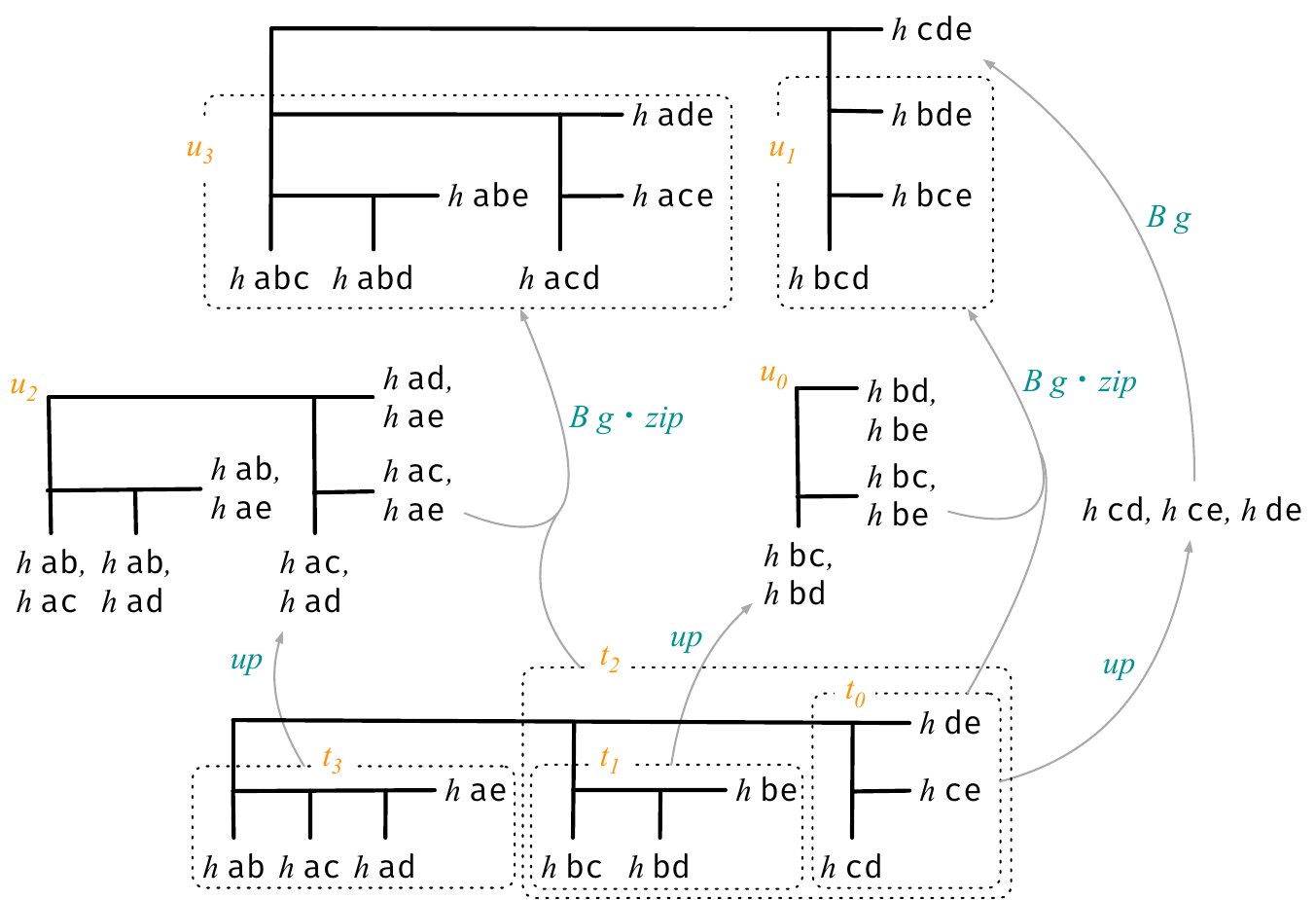}
\caption{Applying \ensuremath{\Varid{B}\;\Varid{g}\;\mathrel{\scalebox{0.6}{$\circ$}}\Varid{up}} to \ensuremath{\Varid{B}\;\Varid{h}\;(\Varid{ch}\;\mathrm{2}\;\mathtt{abcde})}. We abbreviate \ensuremath{\Varid{zipBW}\;\Varid{snoc}} to \ensuremath{\Varid{zip}}.}
\label{fig:up-2-3-demo}
\end{figure}

\paragraph*{An Example.}~~
To demonstrate how \ensuremath{\Varid{up}} works, shown at the bottom of Figure~\ref{fig:up-2-3-demo} is the tree built by \ensuremath{\Varid{B}\;\Varid{h}\;(\Varid{ch}\;\mathrm{2}\;\mathtt{abcde})}.
If we apply \ensuremath{\Varid{up}} to this tree,
the fourth clause of \ensuremath{\Varid{up}} is matched, and we traverse along its right spine until reaching \ensuremath{\Varid{t}_{0}},
which matches the second clause of \ensuremath{\Varid{up}}, and a singleton tree containing \ensuremath{[\mskip1.5mu \Varid{h}\;\mathtt{cd},\Varid{h}\;\mathtt{ce},\Varid{h}\;\mathtt{de}\mskip1.5mu]} is generated.

Traversing backwards, \ensuremath{\Varid{up}\;\Varid{t}_{1}} generates \ensuremath{\Varid{u}_{0}}, which shall have the same shape as \ensuremath{\Varid{t}_{0}} and can be zipped together to form \ensuremath{\Varid{u}_{1}}. Similarly, \ensuremath{\Varid{up}\;\Varid{t}_{3}} generates \ensuremath{\Varid{u}_{2}}, which shall have the same shape as \ensuremath{\Varid{t}_{2}}. Zipping them together, we get \ensuremath{\Varid{u}_{3}}. They constitute \ensuremath{\Varid{B}\;\Varid{h}\;(\Varid{ch}\;\mathrm{3}\;\mathtt{abcde})}, shown at the top of Figure~\ref{fig:up-2-3-demo}.

\subsection{Interlude: Shape Constraints with Dependent Types}
\label{sec:deptypes}

While the derivation guarantees that the function \ensuremath{\Varid{up}}, as defined above, satisfies \eqref{eq:up-spec-B}, the partiality of \ensuremath{\Varid{up}} still makes one uneasy.
Why is it that \ensuremath{\Varid{up}\;\Varid{t}} in the second clause always returns a \ensuremath{\Conid{T}}?
What guarantees that \ensuremath{\Varid{up}\;\Varid{t}} and \ensuremath{\Varid{u}} in the last clause always have the same shape and can be zipped together?
In this section we try to gain more understanding of the tree construction with the help of dependent types.


Certainly, \ensuremath{\Varid{ch}} does not generate all trees of type \ensuremath{\Conid{B}}, but only those trees having certain shapes.
We can talk about the shapes formally by annotating \ensuremath{\Conid{B}} with indices, as in the following Agda datatype:
\begin{hscode}\SaveRestoreHook
\column{B}{@{}>{\hspre}l<{\hspost}@{}}%
\column{3}{@{}>{\hspre}l<{\hspost}@{}}%
\column{7}{@{}>{\hspre}l<{\hspost}@{}}%
\column{E}{@{}>{\hspre}l<{\hspost}@{}}%
\>[B]{}\mathbf{data}\;\Conid{B}\;(\Varid{a}\mathbin{:}\Conid{Set})\mathbin{:}\mathbb{N}\to \mathbb{N}\to \Conid{Set}\;\mathbf{where}{}\<[E]%
\\
\>[B]{}\hsindent{3}{}\<[3]%
\>[3]{}\Conid{T}_{0}{}\<[7]%
\>[7]{}\mathbin{:}\Varid{a}\to \Conid{B}\;\Varid{a}\;\mathrm{0}\;\Varid{n}{}\<[E]%
\\
\>[B]{}\hsindent{3}{}\<[3]%
\>[3]{}\Conid{T}_{\Varid{n}}{}\<[7]%
\>[7]{}\mathbin{:}\Varid{a}\to \Conid{B}\;\Varid{a}\;(\Varid{suc}\;\Varid{n})\;(\Varid{suc}\;\Varid{n}){}\<[E]%
\\
\>[B]{}\hsindent{3}{}\<[3]%
\>[3]{}\Conid{N}{}\<[7]%
\>[7]{}\mathbin{:}\Conid{B}\;\Varid{a}\;\Varid{k}\;\Varid{n}\to \Conid{B}\;\Varid{a}\;(\Varid{suc}\;\Varid{k})\;\Varid{n}\to \Conid{B}\;\Varid{a}\;(\Varid{suc}\;\Varid{k})\;(\Varid{suc}\;\Varid{n})~~.{}\<[E]%
\ColumnHook
\end{hscode}\resethooks
The intention is that \ensuremath{\Conid{B}\;\Varid{a}\;\Varid{k}\;\Varid{n}} is the tree representing choosing \ensuremath{\Varid{k}} elements from a list of length \ensuremath{\Varid{n}}.
Notice that the changes of indices in \ensuremath{\Conid{B}} follow the definition of \ensuremath{\Varid{ch}}.
We now have two base cases, \ensuremath{\Conid{T}_{0}} and \ensuremath{\Conid{T}_{\Varid{n}}}, corresponding to choosing \ensuremath{\mathrm{0}} elements and all elements from a list.
A tree \ensuremath{\Conid{N}\;\Varid{t}\;\Varid{u}\mathbin{:}\Conid{B}\;\Varid{a}\;(\mathrm{1}\mathbin{+}\Varid{k})\;(\mathrm{1}\mathbin{+}\Varid{n})} represents choosing \ensuremath{\mathrm{1}\mathbin{+}\Varid{k}} elements from a list of length \ensuremath{\mathrm{1}\mathbin{+}\Varid{n}}, and the two ways to do so are \ensuremath{\Varid{t}\mathbin{:}\Conid{B}\;\Varid{a}\;\Varid{k}\;\Varid{n}} (choosing \ensuremath{\Varid{k}} from \ensuremath{\Varid{n}}) and \ensuremath{\Varid{u}\mathbin{:}\Conid{B}\;\Varid{a}\;(\mathrm{1}\mathbin{+}\Varid{k})\;\Varid{n}} (choosing \ensuremath{\mathrm{1}\mathbin{+}\Varid{k}} from \ensuremath{\Varid{n}}).
With the definition, \ensuremath{\Varid{ch}} may have type
\begin{hscode}\SaveRestoreHook
\column{B}{@{}>{\hspre}l<{\hspost}@{}}%
\column{E}{@{}>{\hspre}l<{\hspost}@{}}%
\>[B]{}\Varid{ch}\mathbin{:}(\Varid{k}\mathbin{:}\mathbb{N})\to \{\mskip1.5mu \Varid{n}\mathbin{:}\mathbb{N}\mskip1.5mu\}\to \Varid{k}\leq \Varid{n}\to \Conid{Vec}\;\Varid{a}\;\Varid{n}\to \Conid{B}\;(\Conid{Vec}\;\Varid{a}\;\Varid{k})\;\Varid{k}\;\Varid{n}~~,{}\<[E]%
\ColumnHook
\end{hscode}\resethooks
where \ensuremath{\Conid{Vec}\;\Varid{a}\;\Varid{n}} denotes a list (vector) of length \ensuremath{\Varid{n}}.

One can see that a pair of \ensuremath{(\Varid{k},\Varid{n})} uniquely determines the shape of the tree.
Furthermore, it can also be proved that if a tree \ensuremath{\Conid{B}\;\Varid{a}\;\Varid{k}\;\Varid{n}} can be built at all, it must be the case that \ensuremath{\Varid{k}\leq \Varid{n}}:
\begin{hscode}\SaveRestoreHook
\column{B}{@{}>{\hspre}l<{\hspost}@{}}%
\column{12}{@{}>{\hspre}l<{\hspost}@{}}%
\column{29}{@{}>{\hspre}l<{\hspost}@{}}%
\column{40}{@{}>{\hspre}l<{\hspost}@{}}%
\column{E}{@{}>{\hspre}l<{\hspost}@{}}%
\>[B]{}\Varid{bounded}{}\<[12]%
\>[12]{}\mathbin{:}\Conid{B}\;\Varid{a}\;\Varid{k}\;\Varid{n}{}\<[29]%
\>[29]{}\to \Varid{k}\leq \Varid{n}{}\<[40]%
\>[40]{}~~.{}\<[E]%
\ColumnHook
\end{hscode}\resethooks
The function \ensuremath{\Varid{unT}_{\Varid{n}}\mathbin{:}\Conid{B}\;\Varid{a}\;(\Varid{suc}\;\Varid{n})\;(\Varid{suc}\;\Varid{n})\to \Varid{a}} extracts the contents stored in a tip, and is only applied when we know, by the type, that the tree must be \ensuremath{\Conid{T}_{\Varid{n}}}.

\begin{figure}
\centering
\newcommand{\dash}{{\text{-}}}
{\small
\begin{hscode}\SaveRestoreHook
\column{B}{@{}>{\hspre}l<{\hspost}@{}}%
\column{6}{@{}>{\hspre}l<{\hspost}@{}}%
\column{7}{@{}>{\hspre}l<{\hspost}@{}}%
\column{9}{@{}>{\hspre}l<{\hspost}@{}}%
\column{10}{@{}>{\hspre}l<{\hspost}@{}}%
\column{16}{@{}>{\hspre}l<{\hspost}@{}}%
\column{24}{@{}>{\hspre}l<{\hspost}@{}}%
\column{30}{@{}>{\hspre}l<{\hspost}@{}}%
\column{36}{@{}>{\hspre}c<{\hspost}@{}}%
\column{36E}{@{}l@{}}%
\column{39}{@{}>{\hspre}l<{\hspost}@{}}%
\column{44}{@{}>{\hspre}l<{\hspost}@{}}%
\column{E}{@{}>{\hspre}l<{\hspost}@{}}%
\>[B]{}\Varid{up}\mathbin{:}\textcolor{Tan}{(\mathrm{0}\mathbin{<}\Varid{k})}\to \textcolor{Tan}{(\Varid{k}\mathbin{<}\Varid{n})}\to \Conid{B}\;\Varid{a}\;\Varid{k}\;\Varid{n}\to \Conid{B}\;(\Conid{Vec}\;\Varid{a}\;(\Varid{suc}\;\Varid{k}))\;(\Varid{suc}\;\Varid{k})\;\Varid{n}{}\<[E]%
\\
\>[B]{}\Varid{up}\;\textcolor{Tan}{\scaleobj{0.8}{\textcolor{Tan}{0{\small <}0}}}\;{}\<[16]%
\>[16]{}\anonymous \;{}\<[30]%
\>[30]{}(\Conid{T}_{0}\;\Varid{x}){}\<[44]%
\>[44]{}\mathrel{=}\textcolor{Tan}{{\textcolor{Tan}{\bot}\dash\Varid{elim}}}\;\textcolor{Tan}{(\scaleobj{0.8}{\textcolor{Tan}{{<}\dash\Varid{irrefl}}}\;\Varid{refl}\;\scaleobj{0.8}{\textcolor{Tan}{0{\small <}0}})}{}\<[E]%
\\
\>[B]{}\Varid{up}\;\anonymous \;{}\<[16]%
\>[16]{}\textcolor{Tan}{\scaleobj{0.8}{\textcolor{Tan}{1{+}n\!<\!1{+}n}}}\;{}\<[30]%
\>[30]{}(\Conid{T}_{\Varid{n}}\;\Varid{x}){}\<[44]%
\>[44]{}\mathrel{=}\textcolor{Tan}{{\textcolor{Tan}{\bot}\dash\Varid{elim}}}\;\textcolor{Tan}{(\scaleobj{0.8}{\textcolor{Tan}{{<}\dash\Varid{irrefl}}}\;\Varid{refl}\;\scaleobj{0.8}{\textcolor{Tan}{1{+}n\!<\!1{+}n}})}{}\<[E]%
\\
\>[B]{}\Varid{up}\;\anonymous \;{}\<[16]%
\>[16]{}\textcolor{Tan}{\scaleobj{0.8}{\textcolor{Tan}{2{+}n\!<\!2{+}n}}}\;{}\<[30]%
\>[30]{}(\Conid{N}\;(\Conid{T}_{\Varid{n}}\;\anonymous )\;\anonymous ){}\<[44]%
\>[44]{}\mathrel{=}\textcolor{Tan}{{\textcolor{Tan}{\bot}\dash\Varid{elim}}}\;\textcolor{Tan}{(\scaleobj{0.8}{\textcolor{Tan}{{<}\dash\Varid{irrefl}}}\;\Varid{refl}\;\scaleobj{0.8}{\textcolor{Tan}{2{+}n\!<\!2{+}n}})}{}\<[E]%
\\[\blanklineskip]%
\>[B]{}\Varid{up}\;\anonymous \;{}\<[7]%
\>[7]{}\anonymous \;{}\<[10]%
\>[10]{}(\Conid{N}\;(\Conid{T}_{0}\;\Varid{p})\;{}\<[24]%
\>[24]{}(\Conid{T}_{\Varid{n}}\;\Varid{q}){}\<[36]%
\>[36]{}){}\<[36E]%
\>[39]{}\mathrel{=}\Conid{T}_{\Varid{n}}\;(\Varid{p}\mathbin{::}\Varid{q}\mathbin{::}[\mskip1.5mu \mskip1.5mu]){}\<[E]%
\\
\>[B]{}\Varid{up}\;\anonymous \;{}\<[7]%
\>[7]{}\anonymous \;{}\<[10]%
\>[10]{}(\Conid{N}\;\Varid{t}\mathord{@}(\Conid{N}\;\anonymous \;\anonymous )\;{}\<[24]%
\>[24]{}(\Conid{T}_{\Varid{n}}\;\Varid{q}){}\<[36]%
\>[36]{}){}\<[36E]%
\>[39]{}\mathrel{=}\Conid{T}_{\Varid{n}}\;(\Varid{snoc}\;(\Varid{unT}_{\Varid{n}}\;(\Varid{up}\;\textcolor{Tan}{(\scaleobj{0.8}{\textcolor{Tan}{\mathsf{s{\leq}s}}}\;\scaleobj{0.8}{\textcolor{Tan}{\mathsf{z{\leq}n}}})}\;\textcolor{Tan}{(\scaleobj{0.8}{\textcolor{Tan}{\mathsf{s{\leq}s}}}\;\scaleobj{0.8}{\textcolor{Tan}{{\leq}\dash\Varid{refl}}})}\;\Varid{t}))\;\Varid{q}){}\<[E]%
\\
\>[B]{}\Varid{up}\;\anonymous \;{}\<[7]%
\>[7]{}\anonymous \;{}\<[10]%
\>[10]{}(\Conid{N}\;(\Conid{T}_{0}\;\Varid{p})\;{}\<[24]%
\>[24]{}\Varid{u}\mathord{@}(\Conid{N}\;\anonymous \;\Varid{u'}){}\<[36]%
\>[36]{}){}\<[36E]%
\>[39]{}\mathrel{=}\Conid{N}\;{}\<[44]%
\>[44]{}(\Varid{B}\;(\lambda \Varid{q}\to \Varid{p}\mathbin{::}\Varid{q}\mathbin{::}[\mskip1.5mu \mskip1.5mu])\;\Varid{u})\;(\Varid{up}\;\textcolor{Tan}{\scaleobj{0.8}{\textcolor{Tan}{{\leq}\dash\Varid{refl}}}}\;\textcolor{Tan}{(\scaleobj{0.8}{\textcolor{Tan}{\mathsf{s{\leq}s}}}\;(\Varid{bounded}\;\Varid{u'}))}\;\Varid{u}){}\<[E]%
\\
\>[B]{}\Varid{up}\;\anonymous \;\textcolor{Tan}{(\scaleobj{0.8}{\textcolor{Tan}{\mathsf{s{\leq}s}}}\;\scaleobj{0.8}{\textcolor{Tan}{1{+}k\!<\!1{+}n}})}\;(\Conid{N}\;\Varid{t}\mathord{@}(\Conid{N}\;\anonymous \;\anonymous )\;\Varid{u}\mathord{@}(\Conid{N}\;\anonymous \;\Varid{u'}))\mathrel{=}{}\<[E]%
\\
\>[B]{}\hsindent{6}{}\<[6]%
\>[6]{}\Conid{N}\;{}\<[9]%
\>[9]{}(\Varid{zipBW}\;\Varid{snoc}\;(\Varid{up}\;\textcolor{Tan}{(\scaleobj{0.8}{\textcolor{Tan}{\mathsf{s{\leq}s}}}\;\scaleobj{0.8}{\textcolor{Tan}{\mathsf{z{\leq}n}}})}\;\textcolor{Tan}{\scaleobj{0.8}{\textcolor{Tan}{1{+}k\!<\!1{+}n}}}\;\Varid{t})\;\Varid{u})\;(\Varid{up}\;\textcolor{Tan}{(\scaleobj{0.8}{\textcolor{Tan}{\mathsf{s{\leq}s}}}\;\scaleobj{0.8}{\textcolor{Tan}{\mathsf{z{\leq}n}}})}\;\textcolor{Tan}{(\scaleobj{0.8}{\textcolor{Tan}{\mathsf{s{\leq}s}}}\;(\Varid{bounded}\;\Varid{u'}))}\;\Varid{u}){}\<[E]%
\ColumnHook
\end{hscode}\resethooks
}
\caption{An Agda implementation of \ensuremath{\Varid{up}}.}
\label{fig:up-agda}
\end{figure}

Figure~\ref{fig:up-agda} shows an Agda implementation of \ensuremath{\Varid{up}}.
The type states that it is defined only for \ensuremath{\mathrm{0}\mathbin{<}\Varid{k}\mathbin{<}\Varid{n}};
the shape of its input tree is determined by \ensuremath{(\Varid{k},\Varid{n})}; the output tree has shape determined by \ensuremath{(\mathrm{1}\mathbin{+}\Varid{k},\Varid{n})}, and the values in the tree are lists of length \ensuremath{\mathrm{1}\mathbin{+}\Varid{k}}.

The first three clauses of \ensuremath{\Varid{up}} eliminate impossible cases.
The remaining four clauses are essentially the same as the non-dependently typed version,
modulo the additional arguments and proof terms, shown in light brown, that are needed to prove that \ensuremath{\Varid{k}} and \ensuremath{\Varid{n}} are within bounds.
In the clause that uses \ensuremath{\Varid{unT}}, the input tree has the form \ensuremath{\Conid{N}\;\Varid{t}\;(\Conid{T}_{\Varid{n}}\;\Varid{q})}.
The right subtree being a \ensuremath{\Conid{T}_{\Varid{n}}} forces the other subtree \ensuremath{\Varid{t}} to have type
\ensuremath{\Conid{B}\;\Varid{a}\;(\mathrm{1}\mathbin{+}\Varid{k})\;(\mathrm{2}\mathbin{+}\Varid{k})} --- the two indices must differ by \ensuremath{\mathrm{1}}. Therefore \ensuremath{\Varid{up}\;\Varid{t}} has type \ensuremath{\Conid{B}\;\Varid{a}\;(\mathrm{2}\mathbin{+}\Varid{k})\;(\mathrm{2}\mathbin{+}\Varid{k})} and must be built by \ensuremath{\Conid{T}_{\Varid{n}}}.
The last clause receives inputs having type \ensuremath{\Conid{B}\;\Varid{a}\;(\mathrm{2}\mathbin{+}\Varid{k})\;(\mathrm{2}\mathbin{+}\Varid{n})}. Both \ensuremath{\Varid{u}} and \ensuremath{\Varid{up}\;\Varid{t}} have types \ensuremath{\Conid{B}\mathbin{...}(\mathrm{2}\mathbin{+}\Varid{k})\;(\mathrm{1}\mathbin{+}\Varid{n})} and, therefore, have the same shape.

\begin{figure}[h]
\centering
\includegraphics[width=0.8\textwidth]{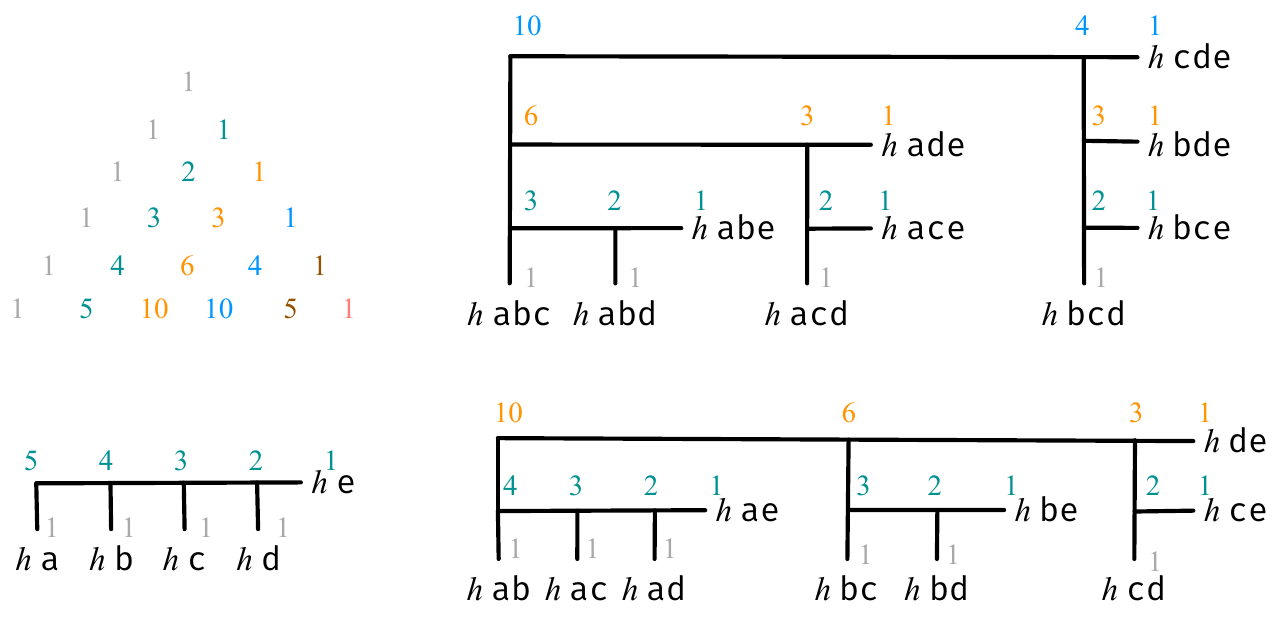}
\caption{Sizes of \ensuremath{\Conid{B}} alone the right spine correspond to diagonals in Pascal's Triangle.}
\label{fig:pascal-tri}
\end{figure}

\paragraph*{Pascal's Triangle.}~~
With so much discussion about choosing, it is perhaps not surprising that the sizes of subtrees along the right spine of a \ensuremath{\Conid{B}} tree correspond to diagonals in Pascal's Triangle.
After all, the \ensuremath{\Varid{k}}-th diagonal (counting from zero) in Pascal's Triangle denotes binomial coefficients --- the numbers of ways to choose \ensuremath{\Varid{k}} elements from \ensuremath{\Varid{k}}, \ensuremath{\Varid{k}\mathbin{+}\mathrm{1}}, \ensuremath{\Varid{k}\mathbin{+}\mathrm{2}}... elements.
This is probably why \cite{Bird:08:Zippy} calls the data structure a \emph{binomial tree}, hence the name \ensuremath{\Conid{B}}.
\footnote{It is not directly related to the tree, having the same name, used in \emph{binomial heaps}.}
See Figure~\ref{fig:pascal-tri} for example.
The sizes along the right spine of \ensuremath{\Varid{ch}\;\mathrm{2}\;\mathtt{abcde}}, that is, \ensuremath{\mathrm{10},\mathrm{6},\mathrm{3},\mathrm{1}}, is the second diagonal (in orange), while the right spine of \ensuremath{\Varid{ch}\;\mathrm{3}\;\mathtt{abcde}} is the fourth diagonal (in blue).
Applying \ensuremath{\Varid{up}} to a tree moves it rightwards and downwards.
In a sense, a \ensuremath{\Conid{B}} tree represents a diagonal in Pascal's Triangle \emph{with a proof} of how it is constructed.

\section{The Bottom-Up Algorithm}

Now that we have constructed an \ensuremath{\Varid{up}} that satisfies \eqref{eq:up-spec-B}, it is time to derive the main algorithm.
Recall that we have defined, in Section~\ref{sec:spec}, \ensuremath{\Varid{h}\;\Varid{xs}\mathrel{=}\Varid{td}\;(\Varid{length}\;\Varid{xs}\mathbin{-}\mathrm{1})\;\Varid{xs}}, where
\begin{hscode}\SaveRestoreHook
\column{B}{@{}>{\hspre}l<{\hspost}@{}}%
\column{11}{@{}>{\hspre}l<{\hspost}@{}}%
\column{E}{@{}>{\hspre}l<{\hspost}@{}}%
\>[B]{}\Varid{td}\mathbin{::}\Conid{Nat}\to \Conid{L}\;\Conid{X}\to \Conid{Y}{}\<[E]%
\\
\>[B]{}\Varid{td}\;\mathrm{0}{}\<[11]%
\>[11]{}\mathrel{=}\Varid{f}\circo\Varid{ex}{}\<[E]%
\\
\>[B]{}\Varid{td}\;(\mathrm{1}\mathbin{+}\Varid{n}){}\<[11]%
\>[11]{}\mathrel{=}\Varid{g}\circo\Varid{L}\;(\Varid{td}\;\Varid{n})\circo\Varid{subs}~~.{}\<[E]%
\ColumnHook
\end{hscode}\resethooks
The intention is that \ensuremath{\Varid{td}\;\Varid{n}} is a function defined for inputs of length exactly \ensuremath{\mathrm{1}\mathbin{+}\Varid{n}}.
We also define a variation:
\begin{hscode}\SaveRestoreHook
\column{B}{@{}>{\hspre}l<{\hspost}@{}}%
\column{12}{@{}>{\hspre}l<{\hspost}@{}}%
\column{E}{@{}>{\hspre}l<{\hspost}@{}}%
\>[B]{}\Varid{td'}\mathbin{::}\Conid{Nat}\to \Conid{L}\;\Conid{Y}\to \Conid{Y}{}\<[E]%
\\
\>[B]{}\Varid{td'}\;\mathrm{0}{}\<[12]%
\>[12]{}\mathrel{=}\Varid{ex}{}\<[E]%
\\
\>[B]{}\Varid{td'}\;(\mathrm{1}\mathbin{+}\Varid{n}){}\<[12]%
\>[12]{}\mathrel{=}\Varid{g}\circo\Varid{L}\;(\Varid{td'}\;\Varid{n})\circo\Varid{subs}~~.{}\<[E]%
\ColumnHook
\end{hscode}\resethooks
The difference is that \ensuremath{\Varid{td'}} calls only \ensuremath{\Varid{ex}} in the base case.
It takes only a routine induction to show that \ensuremath{\Varid{td}\;\Varid{n}\mathrel{=}\Varid{td'}\;\Varid{n}\circo\Varid{L}\;\Varid{f}}.
All the calls to \ensuremath{\Varid{f}} are thus factored to the beginning of the algorithm.
We will then be focusing on transforming \ensuremath{\Varid{td'}}.

Our aim is to show that \ensuremath{\Varid{td'}\;\Varid{n}} can be performed by \ensuremath{\Varid{n}} steps of \ensuremath{\Varid{B}\;\Varid{g}\circo\Varid{up}}, plus some pre and post processing.
Our derivation, however, has to introduce the last step (that is, the leftmost \ensuremath{\Varid{B}\;\Varid{g}\circo\Varid{up}}, when the steps are composed together) separately from the other steps.
We mentioned that \ensuremath{\Varid{subs}} is a special case of \ensuremath{\Varid{choose}}. To be more precise, for \ensuremath{\Varid{xs}} such that \ensuremath{\Varid{length}\;\Varid{xs}\mathrel{=}\mathrm{1}\mathbin{+}\Varid{n}} we have
\begin{align}
\label{eq:unT-up-choose}
\ensuremath{\Varid{subs}\;\Varid{xs}\mathrel{=}\Varid{unT}\circo\Varid{up}\circo\Varid{ch}\;\Varid{n}\myapply\Varid{xs}~~.}
\end{align}
For an example of \eqref{eq:unT-up-choose}, let \ensuremath{\Varid{xs}\mathrel{=}\mathtt{abcd}}. The LHS gives us \ensuremath{[\mskip1.5mu \mathtt{abc},\mathtt{abd},\mathtt{acd},\mathtt{bcd}\mskip1.5mu]}, while
in the RHS, \ensuremath{\Varid{ch}\;\Varid{n}} builds a tree with four tips, which will be joined by \ensuremath{\Varid{up}} to a singleton tree \ensuremath{\Conid{T}\;[\mskip1.5mu \mathtt{abc},\mathtt{abd},\mathtt{acd},\mathtt{bcd}\mskip1.5mu]}.
That \ensuremath{\Varid{up}} always returns a \ensuremath{\Conid{T}} can be seen from the annotated types discussed in Section \ref{sec:deptypes}: since \ensuremath{\Varid{ch}\;\Varid{n}} yields a tree having type \ensuremath{\Conid{B}\;\Varid{a}\;\Varid{n}\;(\mathrm{1}\mathbin{+}\Varid{n})}, \ensuremath{\Varid{up}} has to construct a tree of type \ensuremath{\Conid{B}\;\Varid{a}\;(\mathrm{1}\mathbin{+}\Varid{n})\;(\mathrm{1}\mathbin{+}\Varid{n})}, which must be a tip.

Now we calculate:
\begin{hscode}\SaveRestoreHook
\column{B}{@{}>{\hspre}l<{\hspost}@{}}%
\column{6}{@{}>{\hspre}l<{\hspost}@{}}%
\column{8}{@{}>{\hspre}l<{\hspost}@{}}%
\column{E}{@{}>{\hspre}l<{\hspost}@{}}%
\>[6]{}\Varid{td}\;(\mathrm{1}\mathbin{+}\Varid{n}){}\<[E]%
\\
\>[B]{}\mathrel{=}{}\<[8]%
\>[8]{}\mbox{\commentbegin  since \ensuremath{\Varid{td}\;\Varid{k}\mathrel{=}\Varid{td'}\;\Varid{k}\circo\Varid{L}\;\Varid{f}}  \commentend}{}\<[E]%
\\
\>[B]{}\hsindent{6}{}\<[6]%
\>[6]{}\Varid{td'}\;(\mathrm{1}\mathbin{+}\Varid{n})\circo\Varid{L}\;\Varid{f}{}\<[E]%
\\
\>[B]{}\mathrel{=}{}\<[8]%
\>[8]{}\mbox{\commentbegin  def. of \ensuremath{\Varid{td'}}  \commentend}{}\<[E]%
\\
\>[B]{}\hsindent{6}{}\<[6]%
\>[6]{}\Varid{g}\circo\Varid{L}\;(\Varid{td'}\;\Varid{n})\circo\Varid{subs}\circo\Varid{L}\;\Varid{f}{}\<[E]%
\\
\>[B]{}\mathrel{=}{}\<[8]%
\>[8]{}\mbox{\commentbegin  by \eqref{eq:unT-up-choose}  \commentend}{}\<[E]%
\\
\>[B]{}\hsindent{6}{}\<[6]%
\>[6]{}\Varid{g}\circo\Varid{L}\;(\Varid{td'}\;\Varid{n})\circo\Varid{unT}\circo\Varid{up}\circo\Varid{ch}\;(\mathrm{1}\mathbin{+}\Varid{n})\circo\Varid{L}\;\Varid{f}{}\<[E]%
\\
\>[B]{}\mathrel{=}{}\<[8]%
\>[8]{}\mbox{\commentbegin  naturality of \ensuremath{\Varid{unT}}  \commentend}{}\<[E]%
\\
\>[B]{}\hsindent{6}{}\<[6]%
\>[6]{}\Varid{unT}\circo\Varid{B}\;(\Varid{g}\circo\Varid{L}\;(\Varid{td'}\;\Varid{n}))\circo\Varid{up}\circo\Varid{ch}\;(\mathrm{1}\mathbin{+}\Varid{n})\circo\Varid{L}\;\Varid{f}{}\<[E]%
\\
\>[B]{}\mathrel{=}{}\<[8]%
\>[8]{}\mbox{\commentbegin  naturality of \ensuremath{\Varid{up}}  \commentend}{}\<[E]%
\\
\>[B]{}\hsindent{6}{}\<[6]%
\>[6]{}\Varid{unT}\circo\Varid{B}\;\Varid{g}\circo\Varid{up}\circo\Varid{B}\;(\Varid{td'}\;\Varid{n})\circo\Varid{ch}\;(\mathrm{1}\mathbin{+}\Varid{n})\circo\Varid{L}\;\Varid{f}~~.{}\<[E]%
\ColumnHook
\end{hscode}\resethooks
That gives us the last \ensuremath{\Varid{B}\;\Varid{g}\circo\Varid{up}}.

For the other steps, the following lemma shows that \ensuremath{\Varid{B}\;(\Varid{td'}\;\Varid{n})\circo\Varid{ch}\;(\mathrm{1}\mathbin{+}\Varid{n})} can be performed by \ensuremath{\Varid{n}} steps of \ensuremath{\Varid{B}\;\Varid{g}\circo\Varid{up}}, after some preprocessing.
This is the key lemma that relates \eqref{eq:up-spec-B} to the main algorithm.
\begin{lemma}\label{lma:main}
\ensuremath{\Varid{B}\;(\Varid{td'}\;\Varid{n})\circo\Varid{ch}\;(\mathrm{1}\mathbin{+}\Varid{n})\mathrel{=}{(\Varid{B}\;\Varid{g}\circo\Varid{up})}^{\Varid{n}}\circo\Varid{B}\;\Varid{ex}\circo\Varid{ch}\;\mathrm{1}}.
\end{lemma}
\begin{proof}
For \ensuremath{\Varid{n}\mathbin{:=}\mathrm{0}} both sides simplify to \ensuremath{\Varid{B}\;\Varid{ex}\circo\Varid{ch}\;\mathrm{1}}. For \ensuremath{\Varid{n}\mathbin{:=}\mathrm{1}\mathbin{+}\Varid{n}}:
\begin{hscode}\SaveRestoreHook
\column{B}{@{}>{\hspre}l<{\hspost}@{}}%
\column{6}{@{}>{\hspre}l<{\hspost}@{}}%
\column{9}{@{}>{\hspre}l<{\hspost}@{}}%
\column{E}{@{}>{\hspre}l<{\hspost}@{}}%
\>[6]{}\Varid{B}\;(\Varid{td'}\;(\mathrm{1}\mathbin{+}\Varid{n}))\circo\Varid{ch}\;(\mathrm{2}\mathbin{+}\Varid{n}){}\<[E]%
\\
\>[B]{}\mathrel{=}{}\<[9]%
\>[9]{}\mbox{\commentbegin  def. of \ensuremath{\Varid{td'}}  \commentend}{}\<[E]%
\\
\>[B]{}\hsindent{6}{}\<[6]%
\>[6]{}\Varid{B}\;(\Varid{g}\circo\Varid{L}\;(\Varid{td'}\;\Varid{n})\circo\Varid{subs})\circo\Varid{ch}\;(\mathrm{2}\mathbin{+}\Varid{n}){}\<[E]%
\\
\>[B]{}\mathrel{=}{}\<[9]%
\>[9]{}\mbox{\commentbegin  by \eqref{eq:up-spec-B}  \commentend}{}\<[E]%
\\
\>[B]{}\hsindent{6}{}\<[6]%
\>[6]{}\Varid{B}\;(\Varid{g}\circo\Varid{L}\;(\Varid{td'}\;\Varid{n}))\circo\Varid{up}\circo\Varid{ch}\;(\mathrm{1}\mathbin{+}\Varid{n}){}\<[E]%
\\
\>[B]{}\mathrel{=}{}\<[9]%
\>[9]{}\mbox{\commentbegin  \ensuremath{\Varid{up}} natural  \commentend}{}\<[E]%
\\
\>[B]{}\hsindent{6}{}\<[6]%
\>[6]{}\Varid{B}\;\Varid{g}\circo\Varid{up}\circo\Varid{B}\;(\Varid{td'}\;\Varid{n})\circo\Varid{ch}\;(\mathrm{1}\mathbin{+}\Varid{n}){}\<[E]%
\\
\>[B]{}\mathrel{=}{}\<[9]%
\>[9]{}\mbox{\commentbegin  induction  \commentend}{}\<[E]%
\\
\>[B]{}\hsindent{6}{}\<[6]%
\>[6]{}\Varid{B}\;\Varid{g}\circo\Varid{up}\circo{(\Varid{B}\;\Varid{g}\circo\Varid{up})}^{\Varid{n}}\circo\Varid{B}\;\Varid{ex}\circo\Varid{ch}\;\mathrm{1}{}\<[E]%
\\
\>[B]{}\mathrel{=}{}\<[9]%
\>[9]{}\mbox{\commentbegin  \ensuremath{(\circo)} associative, def. of \ensuremath{{\Varid{f}}^{\Varid{n}}}  \commentend}{}\<[E]%
\\
\>[B]{}\hsindent{6}{}\<[6]%
\>[6]{}{(\Varid{B}\;\Varid{g}\circo\Varid{up})}^{\mathrm{1}\mathbin{+}\Varid{n}}\circo\Varid{B}\;\Varid{ex}\circo\Varid{ch}\;\mathrm{1}~~.{}\<[E]%
\ColumnHook
\end{hscode}\resethooks
\end{proof}

In summary, we have shown that:
\begin{theorem} For all \ensuremath{\Varid{n}\mathbin{::}\Conid{Nat}} we have \ensuremath{\Varid{td}\;\Varid{n}\mathrel{=}\Varid{bu}\;\Varid{n}}, where
\begin{hscode}\SaveRestoreHook
\column{B}{@{}>{\hspre}l<{\hspost}@{}}%
\column{E}{@{}>{\hspre}l<{\hspost}@{}}%
\>[B]{}\Varid{bu}\;\Varid{n}\mathrel{=}\Varid{unT}\circo{(\Varid{B}\;\Varid{g}\circo\Varid{up})}^{\Varid{n}}\circo\Varid{B}\;\Varid{ex}\circo\Varid{ch}\;\mathrm{1}\circo\Varid{L}\;\Varid{f}~~.{}\<[E]%
\ColumnHook
\end{hscode}\resethooks
\end{theorem}
\noindent That is, the top-down algorithm \ensuremath{\Varid{td}\;\Varid{n}} is equivalent to a bottom-up algorithm \ensuremath{\Varid{bu}\;\Varid{n}}, where the input is preprocessed by \ensuremath{\Varid{B}\;\Varid{ex}\circo\Varid{ch}\;\mathrm{1}\circo\Varid{L}\;\Varid{f}}, followed by \ensuremath{\Varid{n}} steps of \ensuremath{\Varid{B}\;\Varid{g}\circo\Varid{up}}. By then we will get a singleton tree, whose content can be extracted by \ensuremath{\Varid{unT}}.
The proof is merely putting all the pieces together.
\begin{proof}
For \ensuremath{\Varid{n}\mathbin{:=}\mathrm{0}}, both sides reduce to \ensuremath{\Varid{f}\circo\Varid{ex}}.
For \ensuremath{\Varid{n}\mathbin{:=}\mathrm{1}\mathbin{+}\Varid{n}}, we have
\begin{hscode}\SaveRestoreHook
\column{B}{@{}>{\hspre}l<{\hspost}@{}}%
\column{4}{@{}>{\hspre}l<{\hspost}@{}}%
\column{9}{@{}>{\hspre}l<{\hspost}@{}}%
\column{E}{@{}>{\hspre}l<{\hspost}@{}}%
\>[4]{}\Varid{td}\;(\mathrm{1}\mathbin{+}\Varid{n}){}\<[E]%
\\
\>[B]{}\mathrel{=}{}\<[9]%
\>[9]{}\mbox{\commentbegin  calculation before  \commentend}{}\<[E]%
\\
\>[B]{}\hsindent{4}{}\<[4]%
\>[4]{}\Varid{unT}\circo\Varid{B}\;\Varid{g}\circo\Varid{up}\circo\Varid{B}\;(\Varid{td'}\;\Varid{n})\circo\Varid{ch}\;(\mathrm{1}\mathbin{+}\Varid{n})\circo\Varid{L}\;\Varid{f}{}\<[E]%
\\
\>[B]{}\mathrel{=}{}\<[9]%
\>[9]{}\mbox{\commentbegin  Lemma~\ref{lma:main}  \commentend}{}\<[E]%
\\
\>[B]{}\hsindent{4}{}\<[4]%
\>[4]{}\Varid{unT}\circo\Varid{B}\;\Varid{g}\circo\Varid{up}\circo{(\Varid{B}\;\Varid{g}\circo\Varid{up})}^{\Varid{n}}\circo\Varid{B}\;\Varid{ex}\circo\Varid{ch}\;\mathrm{1}\circo\Varid{L}\;\Varid{f}{}\<[E]%
\\
\>[B]{}\mathrel{=}{}\<[9]%
\>[9]{}\mbox{\commentbegin  \ensuremath{(\circo)} associative, def. of \ensuremath{{\Varid{f}}^{\Varid{n}}}  \commentend}{}\<[E]%
\\
\>[B]{}\hsindent{4}{}\<[4]%
\>[4]{}\Varid{unT}\circo{(\Varid{B}\;\Varid{g}\circo\Varid{up})}^{\mathrm{1}\mathbin{+}\Varid{n}}\circo\Varid{B}\;\Varid{ex}\circo\Varid{ch}\;\mathrm{1}\circo\Varid{L}\;\Varid{f}{}\<[E]%
\\
\>[B]{}\mathrel{=}{}\<[9]%
\>[9]{}\mbox{\commentbegin  definition of \ensuremath{\Varid{bu}}  \commentend}{}\<[E]%
\\
\>[B]{}\hsindent{4}{}\<[4]%
\>[4]{}\Varid{bu}\;(\mathrm{1}\mathbin{+}\Varid{n})~~.{}\<[E]%
\ColumnHook
\end{hscode}\resethooks
\end{proof}

\section{Conclusion and Discussions}

The sublists problem was one of the examples of \cite{BirdHinze:03:Trouble}, a study of memoisation of functions, with a twist: the memo table is structured according to the call graph of the function, using trees of shared nodes (which they called \emph{nexuses}).
To solve the sublists problem, \cite{BirdHinze:03:Trouble} introduced a data structure, also called a ``binomial tree''. Whereas the binomial tree in~\cite{Bird:08:Zippy} and in this pearl models the structure of the function \ensuremath{\Varid{choose}}, that in \cite{BirdHinze:03:Trouble} can be said to model the function computing \emph{all} sublists:
\begin{hscode}\SaveRestoreHook
\column{B}{@{}>{\hspre}l<{\hspost}@{}}%
\column{18}{@{}>{\hspre}l<{\hspost}@{}}%
\column{E}{@{}>{\hspre}l<{\hspost}@{}}%
\>[B]{}\Varid{sublists}\;[\mskip1.5mu \mskip1.5mu]{}\<[18]%
\>[18]{}\mathrel{=}[\mskip1.5mu [\mskip1.5mu \mskip1.5mu]\mskip1.5mu]{}\<[E]%
\\
\>[B]{}\Varid{sublists}\;(\Varid{x}\mathbin{:}\Varid{xs}){}\<[18]%
\>[18]{}\mathrel{=}\Varid{map}\;(\Varid{x}\mathbin{:})\;(\Varid{sublists}\;\Varid{xs})\mathbin{{+}\mskip-8mu{+}}\Varid{sublists}\;\Varid{xs}~~.{}\<[E]%
\ColumnHook
\end{hscode}\resethooks
Such trees were then extended with up links (and became \emph{nexuses}). Trees were built in a top-down manner, creating carefully maintained links going up and down.

Bird then went on to study the relationship between top-down and bottom-up algorithms, and the sublists problem was one of the examples in \cite{Bird:08:Zippy} to be solved bottom-up. In \cite{Bird:08:Zippy}, a generic top-down algorithm is defined by:
\begin{hscode}\SaveRestoreHook
\column{B}{@{}>{\hspre}l<{\hspost}@{}}%
\column{E}{@{}>{\hspre}l<{\hspost}@{}}%
\>[B]{}\Varid{td}\mathbin{::}\Conid{L}\;\Conid{X}\to \Conid{Y}{}\<[E]%
\\
\>[B]{}\Varid{td}\;\Varid{xs}\mathrel{=}\mathbf{if}\;\Varid{sg}\;\Varid{xs}\;\mathbf{then}\;\Varid{f}\;(\Varid{ex}\;\Varid{xs})\;\mathbf{else}\;(\Varid{g}\circo\Varid{F}\;\Varid{td}\circo\Varid{dc}\myapply\Varid{xs})~~.{}\<[E]%
\ColumnHook
\end{hscode}\resethooks
In his setting, \ensuremath{\Conid{L}} is some input data structure that is often a list in examples, but need not be so. The function \ensuremath{\Varid{sg}\mathbin{::}\Conid{L}\;\Varid{a}\to \Conid{Bool}} determines whether an \ensuremath{\Conid{L}} structure is a singleton, whose content can be extracted by \ensuremath{\Varid{ex}\mathbin{::}\Conid{L}\;\Varid{a}\to \Varid{a}}.
The function \ensuremath{\Varid{dc}\mathbin{::}\Conid{L}\;\Varid{a}\to \Conid{F}\;(\Conid{L}\;\Varid{a})} decomposes an \ensuremath{\Conid{L}} into an \ensuremath{\Conid{F}} structure of \ensuremath{\Conid{L}}s, to be recursively processed.
In the simplest example, \ensuremath{\Conid{L}} is the type of lists, \ensuremath{\Conid{F}\;\Varid{a}\mathrel{=}(\Varid{a},\Varid{a})}, and \ensuremath{\Varid{dc}\;\Varid{xs}\mathrel{=}(\Varid{init}\;\Varid{xs},\Varid{tail}\;\Varid{xs})} (e.g. \ensuremath{\Varid{dc}\;\mathtt{abcd}\mathrel{=}(\mathtt{abc},\mathtt{bcd})}).

A simplified version of Bird's generic bottom-up algorithm, without the nexus, is something like: \ensuremath{\Varid{bu}\mathrel{=}\Varid{ex}\circo{(\Varid{L}\;\Varid{g}\circo\Varid{cd})}^{*}\circo\Varid{L}\;\Varid{f}}.
The pre and postprocessing are respectively \ensuremath{\Varid{L}\;\Varid{f}} and \ensuremath{\Varid{ex}}, while \ensuremath{\Varid{L}\;\Varid{g}\circo\Varid{cd}} is repeatedly performed (via the "\ensuremath{\mathbin{*}}") until we have a singleton.
The function \ensuremath{\Varid{cd}\mathbin{::}\Conid{L}\;\Varid{a}\to \Conid{L}\;(\Conid{F}\;\Varid{a})} transforms one level to the next.
Note that its return type is symmetric to that of \ensuremath{\Varid{dc}} (hence the name \ensuremath{\Varid{cd}}, probably). For the \ensuremath{\Varid{dc}} above, we let \ensuremath{\Varid{cd}} be the function that combines adjacent elements of a list into pairs, e.g., \ensuremath{\Varid{cd}\;\mathtt{abcd}\mathrel{=}[\mskip1.5mu (\mathtt{a},\Varid{b}),(\mathtt{b},\mathtt{c}),(\mathtt{c},\mathtt{d})\mskip1.5mu]}.

While diagrams such as Figure~\ref{fig:ch-lattice} may help one to see how a bottom-up algorithm works, to understand how a top-down algorithm is transformed to a bottom-up one,
it may be it more helpful to think in terms of right-to-left function composition.
Bird's top-down algorithm, when expanded, has the form

\begin{equation}
\label{eq:td-expanded}
 \ensuremath{\Varid{g}\circo\Varid{F}\;\Varid{g}~\mathbin{...}~{\Varid{F}}^{n-1}\;\Varid{g}\circo\textcolor{orange}{{\Varid{F}}^{n}}\;(\Varid{f}\circo\textcolor{orange}{\Varid{ex}})\circo\textcolor{orange}{{\Varid{F}}^{n-1}\;\Varid{dc}}~\mathbin{...}~{\Varid{F}}^2\;\Varid{dc}\circo\Varid{F}\;\Varid{dc}\circo\Varid{dc}~~.}
\end{equation}
Several crucial properties are needed to turn \ensuremath{\Varid{td}} into \ensuremath{\Varid{bu}}.
Among them, Bird needed \ensuremath{\Varid{F}\;\Varid{ex}\circo\Varid{dc}\mathrel{=}\Varid{ex}\circo\Varid{cd}} and \ensuremath{\Varid{dc}\circo\Varid{cd}\mathrel{=}\Varid{F}\;\Varid{cd}\circo\Varid{dc}}.
The first property turns the inner \ensuremath{\textcolor{orange}{{\Varid{F}}^{n}\;\Varid{ex}\circo{\Varid{F}}^{n-1}\;\Varid{dc}}} into \ensuremath{{\Varid{F}}^{n-1}\;(\Varid{ex}\circo\Varid{cd})}, while the second swaps \ensuremath{\Varid{cd}} to the rightmost position.
Function calls to \ensuremath{\Varid{f}} and \ensuremath{\Varid{g}} are shunted to the right by naturalty.
That yields \emph{one} \ensuremath{\Varid{L}\;\Varid{g}\circo\Varid{cd}}.
The process needs to be repeated to create more \ensuremath{\Varid{L}\;\Varid{g}\circo\Varid{cd}}.
Therefore Bird used two inductive proofs to show that \ensuremath{\Varid{td}\mathrel{=}\Varid{bu}}.

The sublists problem, however, does not fit into this framework very well.
While \ensuremath{\Varid{dc}} (which is our \ensuremath{\Varid{subs}}) has type \ensuremath{\Conid{L}\;\Varid{a}\to \Conid{L}\;(\Conid{L}\;\Varid{a})} in the specification,
Bird noticed that we need binomial trees to enable the bottom-up construction, therefore \ensuremath{\Varid{cd}} (our \ensuremath{\Varid{up}}) has type \ensuremath{\Conid{B}\;\Varid{a}\to \Conid{B}\;(\Conid{L}\;\Varid{a})}.
Rather than constructing \ensuremath{\Varid{cd}} from a specification having \ensuremath{\Varid{dc}},
Bird introduced \ensuremath{\Varid{cd}} out of the blue, before introducing another equally cryptic \ensuremath{\Varid{dc'}\mathbin{::}\Conid{B}\;\Varid{a}\to \Conid{L}\;(\Conid{B}\;\Varid{a})} and claiming that \ensuremath{\Varid{dc'}\circo\Varid{cd}\mathrel{=}\Varid{F}\;\Varid{cd}\circo\Varid{dc'}}.

In this pearl we reviewed this problem from the basics,
and instead proposed \eqref{eq:up-spec-B} as a specification of \ensuremath{\Varid{up}}, as well as the property that drives the entire derivation.
Look at the expanded top-down algorithm again:
\begin{equation*}
 \ensuremath{\Varid{g}\circo\Varid{B}\;\Varid{g}~\mathbin{...}~{\Varid{B}}^{n-1}\;\Varid{g}\circo{\Varid{B}}^{n}\;(\Varid{f}\circo\Varid{ex})\circo{\Varid{B}}^{n-1}\;\Varid{subs}~\mathbin{...}~{\Varid{B}}^2\;\Varid{subs}\circo\textcolor{orange}{\Varid{B}\;\Varid{subs}\circo\Varid{subs}}~~,}
\end{equation*}
(The above is \eqref{eq:td-expanded} with \ensuremath{(\Varid{F},\Varid{dc})\mathbin{:=}(\Varid{B},\Varid{subs})}.)
Property \eqref{eq:up-spec-B} turns the \emph{outermost}
\ensuremath{\textcolor{orange}{\Varid{B}\;\Varid{subs}\circo\Varid{subs}}}, which is \ensuremath{\Varid{B}\;\Varid{sub}\circo\Varid{ch}\;(\Varid{n}\mathbin{-}\mathrm{1})}, into \ensuremath{\Varid{up}\circo\Varid{ch}\;(\Varid{n}\mathbin{-}\mathrm{2})}.
That is, \ensuremath{\Varid{up}} is generated from the outside, before being shunted leftwards using naturality.
This fits the problem better: we do not need a \ensuremath{\Varid{dc'}\mathbin{::}\Conid{B}\;\Varid{a}\to \Conid{L}\;(\Conid{B}\;\Varid{a})},
and we need only one inductive proof.

The moral of this story is that while many bottom-up algorithms look alike --- they all have the form \ensuremath{\Varid{post}\circo{\Varid{step}}^{*}\circo\Varid{pre}}, the reason why they work could be very different.
It is likely that there are more patterns yet to be discovered.

\paragraph*{Acknowledgements}~ The author would like to thank Hsiang-Shang Ko and Jeremy Gibbons for many in-depth discussions throughout the development of this work, Conor McBride for discussion at a WG 2.1 meeting, and Yu-Hsuan Wu and Chung-Yu Cheng for proof-reading drafts of this pearl.
The examples of how the immediate sublists problem may be put to work was suggested by Meng-Tsung Tsai.



\end{document}